\documentclass[a4paper,UKenglish,cleveref]{lipics-v2021}

\hideLIPIcs  
\nolinenumbers

\usepackage{xcolor}
\usepackage{tikz}
\usetikzlibrary{shapes.geometric}
\usetikzlibrary{decorations}
\usepackage{mathtools}   
\usepackage{bm}          
\usepackage{todonotes}

\newtheorem{problem}{Problem}
\crefname{problem}{Problem}{Problems}
\Crefname{problem}{Problem}{Problems}

\title{Rich Sequences and Decidability of Arithmetic Theories}

\author{Toghrul Karimov}{Max Planck Institute for Software Systems, Germany}{toghs@mpi-sws.org}{https://orcid.org/0000-0002-9405-2332}{}

\author{Joris Nieuwveld}{Oxford University, UK}{joris.nieuwveld@cs.ox.ac.uk}{https://orcid.org/0009-0002-0339-1230}{}

\author{Jo\"el Ouaknine}{Max Planck Institute for Software Systems, Germany}{joel@mpi-sws.org}{https://orcid.org/0000-0003-0031-9356}{}

\authorrunning{T. Karimov, J. Nieuwveld, and J. Ouaknine}

\Copyright{Toghrul Karimov, Joris Nieuwveld, and Jo\"el Ouaknine} 



\category{} 

\relatedversion{} 

\acknowledgements{Toghrul Karimov and Jo{\"e}l Ouaknine were supported by the DFG grant 389792660 as part of TRR 248 (see \url{perspicuous-computing.science}). 
Jo{\"e}l Ouaknine is also affiliated with Keble College, Oxford as \url{emmy.network} Fellow. 
Joris Nieuwveld was supported by the Glasstone Benefaction, University of Oxford [Violette and Samuel Glasstone Research Fellowships in Science 2025]}

\EventEditors{John Q. Open and Joan R. Access}
\EventNoEds{2}
\EventLongTitle{42nd Conference on Very Important Topics (CVIT 2016)}
\EventShortTitle{CVIT 2016}
\EventAcronym{CVIT}
\EventYear{2016}
\EventDate{December 24--27, 2016}
\EventLocation{Little Whinging, United Kingdom}
\EventLogo{}
\SeriesVolume{42}
\ArticleNo{23}

\newcommand{\frm}[1]{\mathsf{#1}}

\newcommand{\Log}{\operatorname{Log}}

\newcommand{\torus}{\mathbb{T}}
\newcommand{\nat}{\mathbb{N}}
\newcommand{\intg}{\mathbb{Z}}
\newcommand{\rel}{\mathbb{R}}
\newcommand{\rat}{\mathbb{Q}}
\newcommand{\com}{\mathbb{C}}
\newcommand{\alg}{\overline{\rat}}
\newcommand{\ralg}{\rel \cap \alg}

\newcommand{\Dcal}{\mathcal{D}}

\newcommand{\Hcal}{\mathcal{H}}
\newcommand{\Ical}{\mathcal{I}}
\newcommand{\Jcal}{\mathcal{J}}

\newcommand{\Mcal}{\mathcal{M}}

\newcommand{\Tcal}{\mathcal{T}}

\newcommand{\Mb}{\mathbb{M}}

\newcommand{\Rea}{\operatorname{Re}}
\newcommand{\Ima}{\operatorname{Im}}
\newcommand{\seq}[1]{(#1)_{n \in \mathbb{N}}}

\newcommand{\im}{\bm{i}}

\DeclareMathOperator{\lpfsymb}{lpf}
\newcommand{\lpf}[1]{\lpfsymb(#1)}

\newcommand{\seqlinefig}[2][]{%
  \begin{tikzpicture}[x=1.1cm, y=1.1cm, line width=0.7pt,
      dsh/.style={dash pattern=on 2pt off 1.2pt, dash expand off},
      every node/.style={inner sep=1.5pt}]
    \def\lab##1{#2}%
    \draw (-0.5,0) -- (10.5,0);
    \foreach \k in {0,...,10} \draw (\k,-0.27) -- (\k,0.27);
    \draw[dsh,blue] (0,0)    -- (-0.2,-0.6) node[below] {$\lab{0}$};
    \draw[dsh]      (2.5,0)  -- (2.5,-0.6)  node[below] {$\lab{1}$};
    \draw[dsh]      (4.5,0)  -- (4.5,-0.6)  node[below] {$\lab{2}$};
    \draw[dsh]      (6.5,0)  -- (6.5,-0.6)  node[below] {$\lab{3}$};
    \draw[dsh,blue] (8.5,0)  -- (8.5,-0.6)  node[below] {$\lab{8}$};
    \draw[dsh]      (10,0)   -- (10.4,-0.6) node[below] {$\zeta\lab{0}$};
    \draw[dsh,red] (1.45,0) -- (1.45,0.4) node[above] {$\lab{6}$};
    \draw[dsh,red] (3.3,0)  -- (3.2,0.4)  node[anchor=south east, xshift=3pt] {$\lab{4}$};
    \draw[dsh,red] (3.65,0) -- (3.75,0.4) node[anchor=south west, xshift=-3pt] {$\lab{7}$};
    \draw[dsh,red] (7.5,0)  -- (7.5,0.4)  node[above] {$\lab{5}$};
    #1%
  \end{tikzpicture}}

\begin{document}

\maketitle

\begin{abstract}
We develop a new framework for proving the undecidability of first-order theories of structures of the form $\langle \nat; +, P \rangle$, $\langle \nat; <, f \rangle$, and $\langle \nat; +, f\rangle$, where $P \subseteq \nat$ and $f \colon \nat \to \nat$.
It is based on the recent proof of Hieronymi and Schulz that the first-order theory of $\langle \nat; +, \{2^n \colon n \in \nat\}, \{3^n \colon n \in \nat\}\rangle$ is undecidable, and capable of transforming various randomness results about integer sequences into undecidability proofs.
We apply our method to a large class of integer linear recurrence sequences, as well as various special functions, in particular showing that the first-order theories of $\langle\nat; +, \{u_n \colon n \in \nat\} \cap \nat\rangle$, $\langle\nat; <, n \mapsto \max\{0,u_n\}\rangle$, and $\langle \nat; <, \phi\rangle$ are  undecidable, where $\seq{u_n}$ is any integer LRS with exactly two non-repeated dominant roots satisfying a non-degeneracy assumption, and $\phi$ is Euler's totient function.
\end{abstract}

\section*{Preface}

This paper is dedicated to the memory of Florian Luca, who recently passed way at the age of 57.
We learned a lot from him, and will remember him dearly.
In the context of this work, Florian was the one who pointed out that our framework likely applies to the Ramanujan tau function, which inspired us to study decidability of first-order theories of various special functions and predicates.

\section{Introduction}

Decidability of various logical theories connected to arithmetic has been a central topic in mathematics and computer science since the formulation of \emph{Hilbert's program}~\cite{zach2007hilbert} in the 1920s, arguably leading to the birth of modern computer science through the works of Turing in the 1930s~\cite{turing1936computable} in the first place.
Hilbert believed that every true mathematical statement must be provable in some formal system using only ``finitary methods''~\cite{hilbert1922logischen}.
His program, however, was proven unattainable by G\"odel's proofs of the \emph{incompleteness theorems}~\cite{godel1931formal}, which established the following: there does not exist an algorithm (in particular, an algorithm that operates on an effectively enumerable set of axioms using finitely many deduction rules) that takes a first-order statement (in a suitable language) and decides whether it is true in the structure $\Mb \coloneqq \langle \nat; 0, 1, <, +, \cdot\rangle$.
Around forty years after G\"odel, an even stronger result was shown by Matiyasevich, Robinson, Davis, and Putnam~\cite{matiyasevich1993}: it is not possible to algorithmically determine whether a given multivariate polynomial $p \in \intg[x_1,\dots,x_d]$ has a zero in $\intg^d$, famously resolving \emph{Hilbert's tenth problem} in the negative.
In particular, the \emph{existential fragment} of the first-order theory of $\Mb$ is undecidable.
These results shaped a fundamental question that has been actively studied through decades into our time: which fragments of the first-order theory of $\Mb$ are decidable?

Presburger~\cite{presburger1929uber} showed already in 1929 that the first-order theory of the structure $\Mb \coloneqq \langle \nat; 0, 1, <, + \rangle$, now called \emph{Presburger arithmetic}, is decidable.\footnote{The inclusion of $<$ and $0,1$ is purely cosmetic: they can be respectively defined by the formulas $\exists z \colon z \ne 0 \land x + z = y$, $\forall y \colon x + y = y$, and $x\ne 0 \land \forall y \colon y < x \Rightarrow y = 0$. 
Similarly, in this work it does not matter much whether we choose $\nat$ or $\intg$ as our domain, as $<$ is always explicitly present.}
Since then, algorithms for deciding various theories connected to Presburger arithmetic have remained an active area of research in theoretical computer science, with deep connections to automata theory~\cite{bruyere1994logic,frougny2010number,charlier2018first,shallit2022logical}, integer linear programming~\cite{chistikov2024integer,hitarth2026optimization,defossez2024integer}, number theory~\cite{xiao2024hilbert, karimov2025decidability, bacik2026variable,bateman1993decidability}, and model theory~\cite{hieronymi2026axiomatizations,conant2025enriching,tong2025distal}.\footnote{The first-order theory of $\langle \nat; \cdot \rangle$, known as \emph{Skolem arithmetic}, is also decidable. However, decidability of first-order theories of various expansions of $\langle \nat; \cdot \rangle$ has not received much attention. We mention that $<$ is definable in $\langle \nat; + \rangle$ but not in $\langle \nat; \cdot \rangle$: in fact, $\langle \nat; <, \cdot\rangle$ defines addition and hence has an undecidable first-order theory~\cite{bes1998undecidable}.}
Our starting point in this work is: which expansions of $\Mb$ retain decidability of the first-order theory?
We will be especially interested in the following question, which is arguably the most fundamental one in the area.
\begin{problem}
	\label{problem-1}
	For which unary predicates $P \subseteq \nat$ and functions $f \colon \nat \to \nat$ are the first-order theories of $\langle \nat; 0,1,<,+,P \rangle$, $\langle \nat; 0,1,<,f\rangle$, and $\langle\nat;0,1,<,+,f\rangle$ decidable?
\end{problem}
Note that the third structure is generally more expressive than the second one or $\langle \nat; 0,1,<,+, \{f(n) \colon n \in \nat\}\rangle$.
Let us first give a quick overview of known decidability and undecidability results relevant to our problem.

\subsection{State of the art for \Cref{problem-1}}

B\"uchi~\cite{buchi1960weak,bes2002survey} showed in 1960 that for any polynomial $p \in \intg[x]$ of degree at least two for which $p(\nat) \cap \nat$ is infinite, the structure $\langle \nat; 0,1,<,+,p(\nat)\cap\nat\rangle$ defines multiplication and hence has an undecidable first-order theory.
Very recently, Xiao~\cite{xiao2024hilbert} showed that $\langle \nat; 0,1,<,+,\{n^2 \colon n \in \nat\}\rangle$ in fact \emph{existentially} defines multiplication, and hence the existential fragment of its first-order theory is undecidable.
In other words, there is no algorithm that can solve any given system of affine inequalities with integer coefficients and unknowns that are required to be perfect squares.
The result of Xiao conjecturally generalises to all power predicates~\cite{pastenvidaux2016}.

B\"uchi~\cite{buchi1960weak} was also the first to link expansions of Presburger arithmetic to automata theory.
For positive integers $k, x$ with $k \ge 2$, denote by $V_k(x)$ the largest (non-negative integer) power of $k$ that divides $x$.\footnote{In this work, when we define a function, we allow ourselves to be vague on finitely many values: in this example, the definition of $V_k(0)$ does not matter for our purposes.}
Then a set $X \subseteq \nat^l$ is definable (by a first-order formula with $l$ free variables) in $\langle \nat; 0,1,<,+,V_k\rangle$ if and only if, when written in base $k$, it forms a regular language $L \subseteq (\{0,\ldots,k-1\}^l)^*$~\cite{buchi1960weak,bruyere1994logic}. This property of $\langle \nat; 0,1,<,+,V_k\rangle$ is now known as being \emph{automatic} with respect to the \emph{numeration system} $\seq{k^n}$.\footnote{A numeration system is just a base in which we perform greedy expansions.}
It follows that the first-order theory of $\langle \nat; 0,1,<,+,V_k\rangle$, which in particular defines the predicate $\{k^n \colon n \in \nat\}$, is decidable.
More generally, a structure $\langle \nat; 0,1,<,+,P_1,\ldots,P_l,f_1,\ldots,f_s\rangle$ is automatic and hence has a decidable first-order theory when there exists a \emph{single} base with respect to which the predicates $P_1,\ldots,P_l$ and the graphs of functions $f_1,\ldots,f_s$, when written in the said base, form a regular language.

What happens when we have predicates or functions that are naturally represented in two different, unrelated bases?
In this case, first of all, by the Cobham--Sem\"enov theorem, the structure cannot be automatic; see~\cite{bes1997undecidable} for the precise statement.
Thus automatic structures are really special.
Villemaire~\cite{villemaire92} showed in 1992 that for any \emph{multiplicatively independent} integers $a,b \ge 2$, meaning that $a^s = b^t \Rightarrow s = t = 0$ for all $s,t \in \nat$, the structure $\langle\nat;0,1,<,+,V_a,V_b\rangle$ in fact defines multiplication and hence has an undecidable first-order theory.\footnote{The existential fragment of this theory for coprime $a,b$ was very recently shown to be decidable in~\cite{nieuwveld2026existential}.}
Recently, Hieronymi and Schulz~\cite{hieronymi2022strong} used a completely novel approach (which is one of the central topics of this paper) to show that the first-order theory of $\langle \nat;0,1,<,+,P_a,P_b\rangle$, where $P_k = \{k^n \colon n \in \nat\}$ for any positive integer $k \ge 2$, is undecidable for any multiplicatively independent $a,b$ as above.\footnote{More precisely, the $\exists^* \forall^* \exists^*$ fragment of the first-order theory of $\langle\nat; 0, 1, <, +, \{2^n \colon n\in\nat\}, \{3^n \colon n\in\nat\}\rangle$ is undecidable, and the existential fragment of the same theory is decidable~\cite{karimov2025decidability}. Decidability thus remains open for the $\exists^* \forall^*$ fragment.}
This was a substantial strengthening of all known Villemaire-type results (e.g.,~\cite{bes1997undecidable}), as well as the Cobham--Sem\"enov theorem.

Automaticity is an elegant way of proving decidability, but when it comes to expansions of $\langle \nat; 0,1,<,+\rangle$ with a single predicate or a function, it is less powerful than \emph{quantifier elimination}, which is also the technique used by Presburger in his original paper. 
Here the state of the art is due to Sem\"enov~\cite{semenov1980certain,semenov1984logical}, who gave a (growth-type) condition called \emph{compatibility with addition}, which is sufficient for the first-order theories of both $\langle \nat; 0,1,<,+, P\rangle$ and $\langle \nat; 0,1,<, +, f \rangle$ to admit quantifier elimination.\footnote{All quantifier elimination results cited in this work apply to extensions of the original structure, e.g.\ by predicates $\varphi_{a,b}(x)\coloneqq x \equiv a \pmod{b}$. 
These extensions do not affect decidability results.}
His criterion yields decidability for, among others, the functions $2^n$, $n!$, and $n \mapsto u_n$ where $\seq{u_n}$ is the Fibonacci sequence.
For structures of the form $\langle \nat;0,1,<,f\rangle$ as well, Sem\"enov~\cite{semenov1980certain} has a criterion for quantifier elimination, called \emph{compatibility with order}, that applies in particular when $f(n+1)-f(n)$ diverges to $+\infty$ effectively, or when $f(n)$ forms an arithmetic sequence. 

There is a third method for proving decidability of first-order theories: finding an effectively computable set of sentences that axiomatise the theory in question.
In the context of \Cref{problem-1}, this method seems to be exactly as powerful as quantifier elimination~\cite{point2000decidable}.

We mention that a lot of recent work has been dedicated to establishing precise complexity bounds for decidable theories.
Three remarkable results of this type are NP-completeness of deciding the
existential theory of $\langle \nat; 0,1,<,+,n\mapsto 2^n\rangle$
\cite{chistikov2024integer}, a 3EXPTIME upper bound for the full first-order
theory of $\langle \nat;0,1,<,+,\{2^n \colon n \in \nat\}\rangle$
\cite{benedikt2023complexity} (in contrast, the first-order theory of
$\langle \nat;0,1,<,+,n\mapsto 2^n\rangle$ is TOWER-complete
\cite{comptonhenson1990}), and NEXPTIME-completeness for the existential
fragment of $\langle\nat;0,1,<,+,\mid\rangle$, i.e.\ Presburger arithmetic with divisibility~\cite{lechner2015complexity,barros2026nexp}.

\subsection{Our approach to \Cref{problem-1}}

We have illustrated that for \Cref{problem-1}, known decidability and undecidability results form small islands, congregated around automaticity, quantifier elimination, axiomatisability, and definability of multiplication.
If we take a ``generic'' predicate $P$ or function $f$, then for any of the three structures of \Cref{problem-1}, the answer to the decidability question will be ``we don't know''.
To address this gap, in this work we instead focus on the novel approach of Hieronymi and Schulz~\cite{hieronymi2022strong}, which, as mentioned before, has been applied only once to show undecidability of the first-order theory of $\Mb_{a,b} \coloneqq\langle \nat; 0,1,<,+,P_a,P_b\rangle$ for any multiplicatively independent $a,b \ge 2$.\footnote{Schulz~\cite{schulz2023undefinability} has shown that the structure $\Mb_{a,b}$ \emph{does not} define multiplication for any $a,b$. 
Hence the novel approach of~\cite{hieronymi2022strong} was necessary.}
We develop a full theory of this method, and show that it can be used to prove undecidability in previously unknown cases at an industrial scale.

For simplicity, let us focus on $P_2, P_3$.
For $x \ge 1$, define $\alpha(x)$ by $2^{\alpha(x)} \le x < 2^{\alpha(x)+1}$, and for $x \ge 1$, $x$ not a power of 2, let $\beta(x) = \alpha(x - 2^{\alpha(x)})$.
That is, $\beta$ calculates the (index of the) second-largest bit in the binary expansion of $x$.
The approach of~\cite{hieronymi2022strong} rests on the following number-theoretic randomness result.
\begin{theorem}[Lemma~3.4 in~\cite{hieronymi2022strong}]
	\label{thm::hs-randomness}
	For every finite sequence $(t_i)_{i=1}^N$ over $\nat$ there exist $a, b \in \nat$ such that
	\begin{equation}
    \label{eq::hs-main}
	    (\beta(3^n)-\beta(3^a))_{n=a+1}^{b-1} \cap [0, \beta(3^b)-\beta(3^a)] = (t_i)_{i=1}^N.
	\end{equation}
\end{theorem}
Here the $\cap$ operation takes a sequence and an interval, and returns the sequence obtained by only keeping the elements that belong to the specified interval.
\Cref{thm::hs-randomness} can be interpreted as follows: there exists a function $\Phi$ with a fixed number of inputs (two in \Cref{thm::hs-randomness}, namely $3^a, 3^b \in \nat$), that can be implemented (in a very specific sense) in $\Mb_{2,3}$, that outputs all possible finite sequences over $\nat$.
\Cref{thm::hs-randomness} is then used to argue that, given any Turing machine $\Tcal$, we can construct a first-order sentence $\varphi$ in the language of $\Mb_{2,3}$ that is true if and only if $\Tcal$ halts.
The formula $\varphi$ is interpreted as ``there exist $3^a$ and $3^b$ such that the corresponding finite sequence defined via \eqref{eq::hs-main} encodes a halting run of $\Tcal$''.

We proceed by first abstracting the method of~\cite{hieronymi2022strong} described above.
We will use counter machines instead of Turing machines,  as the former are much simpler.
A critical question is: what kind of function $\Phi$ is allowed?
In the discussion above, we were deliberately vague about what $\Phi$ being first-order implementable in the structure means: it turns out to be something specific that does not match any well-known property, most notably being definable in $\Mb_{2,3}$ in the classical sense.\footnote{The classical definability of $\Phi$ in $\Mb_{2,3}$ requires existence of a formula $\varphi$ in the language of $\Mb_{2,3}$ such that $\varphi(3^a,3^b,n,x)$ holds if and only if $x$ is the $n$th element in the sequence defined by $3^a$ and $3^b$; this is far too strong a requirement.}
In the end, we arrive at a notion of \emph{a structure simulating counter machines}, which implies undecidability of the attendant first-order theory.

\paragraph*{Counter machines} A \emph{$k$-counter machine} $\Mcal$ consists of counters $c_1,\ldots,c_k$ taking positive integer values\footnote{Classically, the counters take values in $\nat$. Our modification is inconsequential, but technically convenient.} and instructions numbered $1,\ldots,H$ for some $H$.
Without loss of generality, we assume that $H > 1$.
The instructions are of the form $\mathsf{INC} \: c_i$, $\mathsf{GOTO} \: l$, $\mathsf{HALT}$, $\mathsf{IF} \: c_i > 1 \: \mathsf{THEN} \: \mathsf{DEC} \: c_i \: \mathsf{ELSE} \: \mathsf{GOTO} \: l$, where~$c_i$ is a counter and $l$ is an instruction number.
That is, the counters can be incremented and decremented, but they cannot go below~1.
The decrement operation also acts like a zero test.
We additionally assume, without loss of generality, that the machine starts with the instruction numbered 1, and has a single $\mathsf{HALT}$ instruction, numbered $H$.
The initial values of the counters are all $1$.
We write $\delta_{\Mcal} \colon \{1,\ldots, H\} \times \nat_{>0}^k \to\{1,\ldots, H\} \times \nat_{>0}^k$ for the (partial) transition function of $\Mcal$ that describes how a configuration consisting of an instruction number $l \ne H$ and the values of the $k$ counters is updated in one step.
By the \emph{trace} of $\Mcal$ we mean the sequence (which can be finite or infinite)
\[
(s_n)_n = (0, \iota_0, c_{1,0}, \ldots, c_{k,0} , 0, \iota_1, c_{1,1}, \ldots, c_{k,1}, \ldots)
\]
such that $s_{n} = 0 \Leftrightarrow n \equiv 0 \, (\bmod \, k+2)$, $\iota_0 = 1$, and $c_{j,0} = 1$ for all $j$,
\[
(\iota_{j+1}, c_{1,j+1},\ldots,c_{k,j+1}) = \delta_\Mcal(\iota_j, c_{1,j},\ldots,c_{k,j})
\]
for all $j$, and $\seq{s_n}$ is finite if and only if it contains a single halting instruction, which moreover occurs in the last block.
Note that we use $0$ as a delimiter between consecutive configurations of $\Mcal$.
The Halting Problem asks to decide whether the execution of a given machine $\Mcal$ ever reaches the $\mathsf{HALT}$ instruction, and is undecidable already for two-counter machines~\cite{minsky1967computation}.

\paragraph*{Structures simulating counter machines}

We say that a structure $\Mb$ with domain $D$ \emph{simulates counter machines} if there exist $l,m \ge 1$, maps $\mathsf{Seq} \colon D^l \to \nat^*$, $\mathsf{Rep} \colon D^l \to (D^m)^*$, and formulas $\frm{rep}$, $\frm{cnst}_k$ (where $k \in \nat$ and we additionally require that $\frm{cnst}_k$ be effectively computable given $k$), $\frm{succ}$, $\frm{inc}$, and $\frm{eq}$ in the language of~$\Mb$ with $l+m$, $l+m$, $l+2m$, $l+2m$, and $l+2m$ free variables, respectively, that satisfy the following.
\begin{enumerate}
	\item The map $\mathsf{Seq}$ is surjective.
	Intuitively, $\mathsf{Seq}$ is a black box with $l$ inputs that is ``implementable'' in $\Mb$ that outputs all possible finite sequences over $\nat$.
	\item For all $x \in D^l$, $|\mathsf{Seq}(x)| = |\mathsf{Rep}(x)|$, and the terms of $\mathsf{Rep}(x)$ are distinct.
	Intuitively, $\mathsf{Rep}(x)$ is a finite sequence of distinct terms that index the sequence $\mathsf{Seq}(x)$.
	\item For all $x \in D^l$ and $y \in D^m$, $\Mb \models \frm{rep}(x,y)$ if and only if  $y= \mathsf{Rep}(x)_i$ for some $i$.
	That is, $y$ appears in $\mathsf{Rep}(x)$; equivalently, $y$ is the index of some term in $\mathsf{Seq}(x)$.
	\item For all $k \in \nat$, $x \in D^l$ and $y \in D^m$, $\Mb \models \frm{cnst}_k(x,y)$ if and only if there exists $i$ such that $y=\mathsf{Rep}(x)_i$ and $\mathsf{Seq}(x)_i = k$:
	i.e., $y$ is the index of a term in $\mathsf{Seq}(x)$ that is equal to~$k$.
	\item For all $x \in D^l$ and
	$y,z \in D^m$, $\Mb \models \frm{succ}(x,y,z)$ if and only if $y,z$ are two consecutive terms appearing in $\mathsf{Rep}(x)$: that is, $y,z$ index two consecutive terms of $\mathsf{Seq}(x)$.
	Note that $\frm{succ}(x,y,z)$ implies $\frm{rep}(x,y)$ and $\frm{rep}(x,z)$.
	\item For all $x \in D^l$ and
	$y,z \in D^m$, $\Mb \models \frm{inc}(x,y,z)$ if and only if  there exist $i,j$ such that $y = \mathsf{Rep}(x)_i$, $z= \mathsf{Rep}(x)_j$, and $\mathsf{Seq}(x)_j = \mathsf{Seq}(x)_i + 1$: that is, $y,z$ index two terms $t_1, t_2$, respectively, of $\mathsf{Seq}(x)$ that satisfy $t_2 = t_1 + 1$.
	\item
	For all $x \in D^l$ and
	$y,z \in D^m$, $\Mb \models \frm{eq}(x,y,z)$ if and only if there exist $i,j$ such that $y = \mathsf{Rep}(x)_i$, $z = \mathsf{Rep}(x)_j$ and $\mathsf{Seq}(x)_j = \mathsf{Seq}(x)_i$: that is, $y,z$ index two terms of $\mathsf{Seq}(x)$ that are equal.
\end{enumerate}

\begin{theorem}[See \Cref{sec::how-to-prove-undec} for the proof]
	\label{thm::how-to-prove-undec}
	Any structure $\Mb$ that simulates counter machines has an undecidable first-order theory.
\end{theorem}

\subsection{First undecidability results}

We now apply our framework and prove our first undecidability result.
We denote by $\phi$ Euler's totient function ($\phi(n)$ is the number of positive integers at most $n$ that are coprime to $n$) and by $s$ the sum of divisors function ($s(n)$ is the sum of all positive divisors of $n$; we set $s(0) = 0$): these two are intimately connected.
Let us proceed by first giving a sufficient condition for a structure to simulate counter machines.

\begin{definition}
    We say that a function $f \colon \nat\to\nat$ realises arbitrary permutations if for every permutation $\sigma \colon \{1,\ldots,k\} \to \{1,\ldots,k\}$ there exists $c \in \nat$ such that
    \begin{equation}
        \label{eq::def-arb-perm-1}
        f(c+\sigma(1)) < \cdots < f(c+\sigma(k)),
    \end{equation}
    i.e., for all $c < n_1,n_2  \le c + k$,
    \begin{equation*}
        \label{eq::def-arb-perm-2}
        f(n_1) < f(n_2) \Leftrightarrow \sigma^{-1}(n_1 - c) < \sigma^{-1}(n_2 - c).
    \end{equation*}
\end{definition}

\begin{lemma}[See \Cref{sec::how-to-prove-undec} for the proof]
    \label{thm::permutation-to-undec}
    Suppose $f$ realises arbitrary permutations.
    Then $\langle \nat; 0, 1, <, f\rangle$ and hence $\langle \nat; 0,1,<,+,f\rangle$ simulate counter machines and have undecidable first-order theories.
\end{lemma}

We give a proof sketch to hopefully illustrate the rationale behind our requirements for $\mathsf{Seq}$ and $\mathsf{Rep}$ in the definition above.
We choose $l=3$ and $m=1$, meaning that our function $\mathsf{Seq}$ that churns out all possible finite sequences will take 3 inputs, and we will index the terms of the output sequences using a single number.
For $c,d,e \in \nat$ we define
	\[
	\mathsf{Rep}(c,d,e) = ( d+1,d+2,\ldots,e)
	\]
	and $\mathsf{Seq}(c,d,e)$ to be the finite sequence $(t_i)_{i=1}^{e-d}$ over $\{0, \ldots, d-c\}$ where
	\[
	t_i =
	\#
	\{c < n \le d \colon f(n) < f(d+i)\}
	.
	\]
Thus $d+i$ indexes the $i$th term of $\mathsf{Seq}(c,d,e)$.
Because $f$ realises arbitrary permutations, it is immediate that $\mathsf{Seq}$ is onto $\nat^*$.
It remains to define $\frm{rep}$, $\frm{cnst}_k$, $\frm{succ}$, $\frm{inc}$, and $\frm{eq}$.
The idea behind all of these is the same: we illustrate it with $\frm{inc}$.
When is the term indexed by $n_2$ the increment (by one) of the term indexed by $n_1$?
When there exists unique $m$ (written $\exists!m$) such that $c < m \le d$ and $f(m) < f(n_2)$ but $f(m) \nless f(n_1)$.
In the end, we obtain
\[
\frm{inc}(c,d,e,n_1,n_2) \coloneqq d+1\le n_1, n_2 \le e \:\land\: \exists ! m \in (c,d]\colon f(m) \in [f(n_1), f(n_2)). \qed
\]

But how do we show that $\phi$ or $s$ realise arbitrary permutations?
It turns out that, more than 70 years ago, the then 18-year-old A. Schinzel already proved a much stronger result!

\begin{theorem}[Theorems 1 and 2 in~\cite{schinzel1955functions}]
    \label{thm::schinzel}
    Let $f$ be either the totient function or the sum of divisors function.
    Then for every $k$, 
    \[
    \bigg(
         \frac{f(n+2)}{f(n+1)}, \frac{f(n+3)}{f(n+2)},\ldots, \frac{f(n+k)}{f(n+k-1)}
    \bigg)_{n \in \nat}
    \]
    is dense in $[0,\infty)^{k-1}$.
\end{theorem}
\begin{corollary}\label{cor:Euler and sum of divs undecidable}
    The first-order theories of $\langle \nat; 0,1,<, \phi\rangle$ and $\langle\nat;0,1,<,s\rangle$ are undecidable.
\end{corollary}
\begin{proof}
    Let $f$ be either of the two functions.
    We will show that it realises arbitrary permutations.
    Fix $\sigma \colon \{1,\ldots,k\} \to \{1,\ldots,k\}$ and let $m = \sigma(1)$.
    By Schinzel's theorem, for any $\varepsilon > 0$ there exists $c$ such that $f(c+i+1)/f(c+i) \in (\sigma^{-1}(i+1)/\sigma^{-1}(i)-\varepsilon, \sigma^{-1}(i+1)/\sigma^{-1}(i)+\varepsilon)$  for all $1 \le i < k$.
    If $\varepsilon$ is sufficiently small, this implies \eqref{eq::def-arb-perm-1}, since $f(c+i)/f(c+m) \approx  \sigma^{-1}(i)/\sigma^{-1}(m) = \sigma^{-1}(i)$, i.e. at most $1/2$ away. Then, 
    \begin{equation*}
        f(c+m) < \frac{3}{2}f(c+m) < f(c+\sigma(2)) < \frac{5}{2}f(c+m) < \cdots < \frac{2k-1}{2}f(c+m) < f(c+\sigma(k)).\qedhere
    \end{equation*}
\end{proof}

This result illustrates our general approach to proving all of our undecidability results. 
We start with a purely mathematical result that says that, in a certain sequence, ``everything that can happen will happen'', e.g.\ we can completely control the ratios of consecutive elements.
We then transform this result into a combinatorial one that allows us to extract arbitrary finite sequences over $\nat$ from finitely many parameters (three in the case of functions that realise arbitrary permutations, discussed above) using the power of first-order logic.

To the best of our knowledge, the totient and sum of divisors functions are the first known natural examples of functions $f$ for which the first-order theory of $\langle\nat;0,1,<,f\rangle$ is undecidable.\footnote{
Of course, if $f$ is allowed to talk about Turing machines (which we exclude in our definition of a ``natural'' function), then undecidability is easy.
For example, take the infinite word $\alpha = \prod_{i=0}^{\infty} v_{\sigma(i)}$ where $\sigma \colon \nat \to (\nat_{>0})^2$ enumerates all possible pairs $(t,k)$ consisting of a Turing machine $t$ (itself represented by a positive integer; 
we assume the initial and halting states are always distinct) and a positive integer~$k$, and $v_{\sigma(i)}$ is $0^k 1^t$ if $t$ halts on $\varepsilon$ exactly after $k$ steps and empty otherwise.
Then the Turing machine $t \in \nat_{>0}$ halts on empty input if and only if $0 1^t 0$ occurs in $\alpha$. 
It follows that the first-order theory of $\langle \nat; <, P \rangle$ is undecidable, where $P$ is the set of all indices at which 1 occurs in $\alpha$.}
We collect all of our undecidability results for special functions and predicates (the least prime factor function, Ramanujan tau function, square-free numbers) in \Cref{sec::special-functions}.

\subsection{A dynamical perspective on \Cref{problem-1}}

In our undecidability results for special functions and predicates, we leverage existing results from number theory to show that a relevant structure simulates counter machines.
In this way we do obtain novelties, but these are still isolated.
To make large-scale progress, we propose to first take a \emph{dynamical perspective} on \Cref{problem-1}.
Writing $ U = \{u_n\colon n\in\nat\} \cap\nat$ and $u(n) = \max \{0,u_n\}$,
we ask: for which integer sequences $\seq{u_n}$ are the first-order theories~of%
\begin{equation}
    \label{eq::structures}
    \langle\nat;0,1,<,+,U\rangle,
    \quad
    \langle\nat;0,1,<,u\rangle,
    \quad\textrm{and}\quad 
\langle\nat; 0,1,<,+,u\rangle
\end{equation}
decidable?
Note that this is just \Cref{problem-1} with more detail.
The advantage of our formulation is that additionally asking what kind of a process is generating $P$ or $f$, rather than focusing solely on intrinsic properties of $P$ and $f$, allows one to apply powerful mathematical tools to settle the relevant decidability questions.
The same phenomenon has recently been demonstrated also in the context of MSO theories of structures of the form $\langle \nat; <, P_1,\ldots,P_m\rangle$, where $P_1,\ldots,P_m$ are unary predicates~\cite{berthe2024decidability}.

Having adopted the dynamical viewpoint, we say that an integer sequence $\seq{u_n}$ is \emph{rich} if any of the three structures in \eqref{eq::structures} simulates counter machines.
So, which classes of integer sequences $\seq{u_n}$ should we consider?
Can we find a large, novel class of rich sequences?

\paragraph*{Linear recurrence sequences}

We focus on the fundamental class of linear recurrence sequences (LRS).
An integer LRS is a sequence $\seq{u_n}$ over $\intg$ such that, for some $d \ge 0$ and $a_1,\ldots,a_d \in \intg$, where $a_d \ne 0$ when $d\ge 1$, we have that 
\begin{equation}
    \label{eq::lrs-def}
    u_{n+d} = a_1 u_{n+d-1} + \cdots + a_d u_n
\end{equation}
for all $n$.\footnote{The restriction $a_d \ne 0$ is not restrictive: every LRS has a suffix that is an LRS of order $d'$ satisfying \eqref{eq::lrs-def} with either $d' = 0$ or $d=d'$ and $a_d \ne 0$.}
The smallest possible value of $d$ is called the \emph{order} of $\seq{u_n}$.
The Fibonacci sequence is an LRS of order two, defined by the recurrence relation $u_{n+2} = u_{n+1} + u_n$.
An LRS of order $d$ satisfying the recurrence relation~\eqref{eq::lrs-def} can be uniquely written in its \emph{exponential-polynomial form}
\begin{equation}
	\label{eq::lrs-1}
	u_n = q_1(n)\lambda_1^n + \cdots+ q_m(n)\lambda_m^n
\end{equation}
where $\lambda_1,\ldots, \lambda_m$ are the distinct roots of the polynomial $p(x) = x^d - a_1x^{d-1} - \cdots - a_d$ and $q_1,\ldots,q_m$ are polynomials (with complex coefficients) with $\sum_{i=1}^m (\deg(q_i)+1) = d$.
The polynomial $p$ is called the \emph{characteristic polynomial}, and the numbers $\lambda_1,\ldots,\lambda_m$ are called the characteristic roots of $\seq{u_n}$.
A root $\lambda_i$ is \emph{dominant} if $|\lambda_i| \ge |\lambda_j|$ for all $j$.
We say that $\seq{u_n}$ is \emph{irreducible} when $p$ is irreducible over $\rat$, i.e.\ not a product of two non-constant polynomials $q, h$ with rational coefficients.

LRS form an expressive class of sequences that has been extensively studied both from the theoretical (e.g., in the context of Skolem and Positivity problems) and the practical perspective (e.g., in program verification).
In the context of arithmetic theories, already at low orders they give rise to wildly different first-order theories.
The decidability status of these is mostly open, but they can be decidable or undecidable for various reasons; see the discussion below about order two LRS.
Finally, another good reason for studying predicates and functions generated by integer LRS is that, in this context, the applicability and limitations of the four classical techniques are well-understood.
\begin{itemize}
    \item The first structure of \eqref{eq::structures} is automatic when $\seq{u_n}$ is \emph{irreducible} and has a single real \emph{dominant root} that is a \emph{Pisot number} \cite{frougny2010number}.\footnote{A real number $\alpha$ with $\alpha>1$ is a Pisot number if it is the root of a monic polynomial $p \in \intg[x]$ whose all other roots $\beta$ satisfy $|\beta| < 1$.}
    This is the only known general result: all other known automatic examples are exceptional, e.g.\ the case of $\seq{u_n}$ being an arithmetic sequence.
    When $\max \{0,u_n\} \notin O(n)$, which is the case for the vast majority of LRS, the second and the third structures of \eqref{eq::structures} provably cannot be automatic with respect to any ``standard'' numeration system (i.e., \emph{linear} system with greedy expansion and a finite digit alphabet)~\cite{allouche2022prove}. 
    \item For the first and the third structures of $\eqref{eq::structures}$ Sem\"enov's criterion (compatibility with addition) for quantifier elimination as well as Point's axiomatisation technique apply when $\seq{u_n}$ is irreducible and has a single, positive real dominant root (that may not be a Pisot number).
    For the second structure, Sem\"enov's relevant criterion for quantifier elimination (compatibility with order) applies when $\seq{u_n}$ has a single real dominant root (that, this time, may be repeated).
    We note that both claims above are not exhaustive: there are exceptional $\seq{u_n}$ that, for example, have a reducible characteristic polynomial, for which $n \mapsto \max\{0,u_n\}$ is compatible with addition.
    \item Undecidability, up until our work, was only known for the first
    and the third structures when $u_n = p(n)$ for some $p \in \intg[x]$ of degree at least two with $p(\nat) \cap \nat$ infinite; in this case too, $\seq{u_n}$ is an LRS with a single dominant root (which is 1), but it is not irreducible.
    Below we show undecidability for the first structure when $u_n = 2^n + 3^n$ (the argument applies to any integer LRS of order two whose roots are positive, multiplicatively independent integers, and $u_n \ge 0$ for infinitely many $n$), which is still within the scope of LRS with one real dominant root.
\end{itemize}

Summarising the discussion above, the decidability of first-order theories of the structures in \eqref{eq::structures}, where $\seq{u_n}$ is an integer LRS, prior to this work was largely (but not fully) understood only when $\seq{u_n}$ has a single, real dominant root $\rho$.
A key property in this case is that $u_n$ is either periodic or $|u_n|$ eventually monotonically diverges to $+\infty$.
LRS with two or more dominant roots, unless they are degenerate (discussed for order two LRS below; see green points in \Cref{fig:order2}), do not exhibit such a simple behaviour, and consequently, nothing general is known about any of the first-order theories of all three structures in \eqref{eq::structures} for such LRS.
Below we illustrate this situation more concretely.

\paragraph*{A close look at LRS of order two}

We now focus on LRS $\seq{u_n}$ of low order and, for brevity, $\Mb = \langle \nat; 0,1,<,+,U \rangle$ where $U = \{u_n \colon n \in\nat\}\cap\nat$.
A similar case analysis can be given for all three structures in \eqref{eq::structures}.

An integer LRS of order one is just a geometric sequence: $u_n = b^n u_0$ for an integer $b \ne 0$, and for such LRS, the first-order theory of $\Mb$ is decidable by Sem\"enov's criterion for quantifier elimination.
At the next level of complexity lie integer LRS of order two, defined by the recurrence relation $u_{n+2} = bu_{n+1} + cu_n$ with $c \ne 0$.
The characteristic polynomial is $p(x) = x^2 - bx - c$, which has two distinct real roots $\rho_1,\rho_2$ when $D \coloneqq b^2+4c > 0$, a repeated real root $\rho$ when $D = 0$, and two conjugate non-real roots $\lambda, \overline{\lambda}$ when $D < 0$.
It is reducible exactly when $D$ is a perfect square, which implies that $D \ge 0$.

Now suppose $D > 0$.
Then $u_n = \alpha\rho_1^n + \beta \rho_2^n$ for $\alpha,\beta,\rho_1,\rho_2 \in \rel$.
The case where $p$ is irreducible (red points in \Cref{fig:order2}) produces $\Mb$ with a decidable first-order theory: in this case, either $\seq{u_n}$ has a single (real) dominant root and Sem\"enov's criterion applies directly, or $\rho_2 = - \rho_1$ and $u_n = \alpha \rho_1^n + \beta((-\rho_1)^{n})$.
In the latter case, $\rho_1^2 =\rho_2^2 = c \in \intg$.
We consider odd and even indices separately: $U = U_0 \cup U_1$ where $U_i = \{u_i(\rho_1^2)^n\ \colon n \in \nat\} \cap \nat$.
Since $\rho_1^2 \in \intg$, the first-order theories of $\langle \nat; 0,1,<,+,U_0,U_1\rangle$ and $\Mb$ are decidable by automaticity.

The case of reducible $p$ (missing points with $c \ne 0$ in \Cref{fig:order2}) is interesting.
In this case, $\rho_1,\rho_2 \in \intg$: by a classical theorem of Gauss, reducibility over $\rat$ implies reducibility over $\intg$.
It is possible to classify whether the first-order theory of $\Mb$ is decidable by a case analysis on the values of $\alpha,\beta,\rho_1,\rho_2$, but we content ourselves with two illustrative examples.
\begin{itemize}
    \item If $b = 5$, $c = -6$, $u_0 = 2$, and $u_1 = 5$, then we have $u_n = 2^n + 3^n$. 
    We can define the predicates $P_2 = \{2^n \colon n \in \nat\}$ and $P_3 = \{3^n \colon n \in \nat\}$ in $\Mb$ by the formulas
    \begin{align*}
        x \in P_2 &\quad\Leftrightarrow\quad \exists y, z \in U \colon \big(
        y < z \:\land\: \forall s\in U \colon (s \ge z \lor s \le y) \:\land\: x+z = 3y
        \big),\\
        x \in P_3 &\quad\Leftrightarrow\quad \exists y, z \in U \colon \big(
        y < z \:\land\: \forall s\in U \colon (s \ge z \lor s \le y) \:\land\: x+2y = z
        \big).
    \end{align*}
    Because 2 and 3 are multiplicatively independent, the result of~\cite{hieronymi2022strong} applies, and the first-order theory of $\Mb$ is undecidable.
    \item The other possibility is that $\rho_1$ and $\rho_2$ are multiplicatively dependent, e.g.\ $\rho_1=4$ and $\rho_2=8$.
    In this case, there exist integers $\rho > 0$, $k_1,k_2 \ge 0$ such that $\rho_1 \in \{\rho^{k_1}, -\rho^{k_1}\}$ and $\rho_2 \in \{ \rho^{k_2}, -\rho^{k_2} \}$.
    Therefore, $U$ is definable in $\langle \nat; 0, 1, <, +, n\mapsto \rho^n\rangle$, which has a decidable first-order theory by Sem\"enov's criterion.
    Hence the first-order theory of $\Mb$ is decidable.
\end{itemize}

Now suppose $D=0$ (black points in \Cref{fig:order2}). 
In this case, $u_n = (\alpha+\beta n) \rho^n$, and to the best of our knowledge, decidability of the first-order theory of $\Mb$ (excluding the trivial cases where $U$ is finite or forms an arithmetic/geometric progression) is open.
A concrete open example is $u_n = n 2^n$, which satisfies the recurrence relation $u_{n+2} = 4u_{n+1} - 4 u_n$.

Finally, suppose $D < 0$, which yields $u_n = a\lambda^n + \overline{a}\,\overline{\lambda}^n$ for non-real $\lambda$.
We say that $\seq{u_n}$ is \emph{degenerate} if $\gamma \coloneqq \lambda / \overline{\lambda}$ is a root of unity, i.e.\ there exists an integer $k > 0$ such that $\gamma^k = 1$ (green points in \Cref{fig:order2}).
In particular, $\lambda^k = \overline{\lambda}^k$, and hence $\lambda^k \in \rat$.\footnote{In fact, because $\lambda^k$ is an \emph{algebraic integer}, we can immediately deduce that $\lambda^k \in \intg$.}
The reason for this name is that, in this case, $\seq{u_n}$ is actually an interleaving of $k$ order-one LRS, rather than a ``genuine'' order two LRS.
Specifically, we have that $U = U_0 \cup \cdots \cup U_{k-1}$ where $U_i = \{u_{nk+i} \colon n \in \nat\} \cap \nat$ and $u_{nk+i} = (a\lambda^i + \overline{a} \, \overline{\lambda}^i)(\lambda^k)^n$.
Since $u_{nk+i} \in \intg$ for all $n$, we deduce that $\ell \coloneqq \lambda^k \in \intg$ and $a\lambda^i + \overline{a} \, \overline{\lambda}^i \in \intg$, and hence every $U_i$ is $|\ell|$-automatic.
It follows that the first-order theories of $\langle \nat;0,1,<,+,U_0,\ldots,U_{k-1}\rangle$ and hence $\Mb$ are decidable.

Only the case of $D < 0$ and $\seq{u_n}$ being non-degenerate (blue points in \Cref{fig:order2}) remains. 
Up until our work, to the best of our knowledge, nothing was known about the corresponding first-order theories.
We show that such $\seq{u_n}$ possess strong randomness properties, and are rich.
Consequently, the first-order theory of $\Mb$ is undecidable.
In the end, for LRS of order two, only the case of a repeated real root ($D=0$) remains open.

\begin{figure}
    \centering
    \begin{tikzpicture}[x=0.44cm,y=0.34cm]
  \draw[->] (-12.5,0) -- (12.5,0) node[right] {$b$};
  \draw[->] (0,-11.5) -- (0,11.5) node[above] {$c$};
  \draw[gray,domain=-6.33:6.33,samples=100] plot (\x,{-\x*\x/4});
  \foreach \b in {-12,...,12}
    \foreach \c in {-10,...,10}{
      \pgfmathtruncatemacro{\SKIP}{(\c==0 ) ? 1 : 0}
      \ifnum\SKIP=1\relax\else
        \pgfmathtruncatemacro{\D}{\b*\b+4*\c}
        \ifnum\D<0
          \pgfmathtruncatemacro{\RU}{(\b==0 || \b*\b+\c==0 || \b*\b+2*\c==0
                                      || \b*\b+3*\c==0) ? 1 : 0}
          \ifnum\RU=1
            \fill[green!60!black] (\b,\c) circle (2.5pt);
          \else
            \fill[blue] (\b,\c) circle (2.5pt);
          \fi
        \else\ifnum\D=0
          \fill[black] (\b,\c) circle (2.5pt);
        \else
          \pgfmathtruncatemacro{\S}{round(sqrt(\D))}
          \ifnum\numexpr\S*\S\relax=\D\relax\else
            \fill[red] (\b,\c) circle (2.5pt);
          \fi
        \fi\fi
      \fi
    }
\end{tikzpicture}
    \caption{Spectrum of $u_{n+2} = bu_{n+1} + cu_n$. 
    Red: two distinct real roots with $p(x) = x^2-bx-c$ irreducible.
    Missing point: either $c = 0$, or $D > 0$ and $p$ is reducible.
    Black: a repeated real root.
    Blue or green: two conjugate non-real roots.
    Green: the ratio of the two roots is a root of unity.%
    }
    \label{fig:order2}
\end{figure}
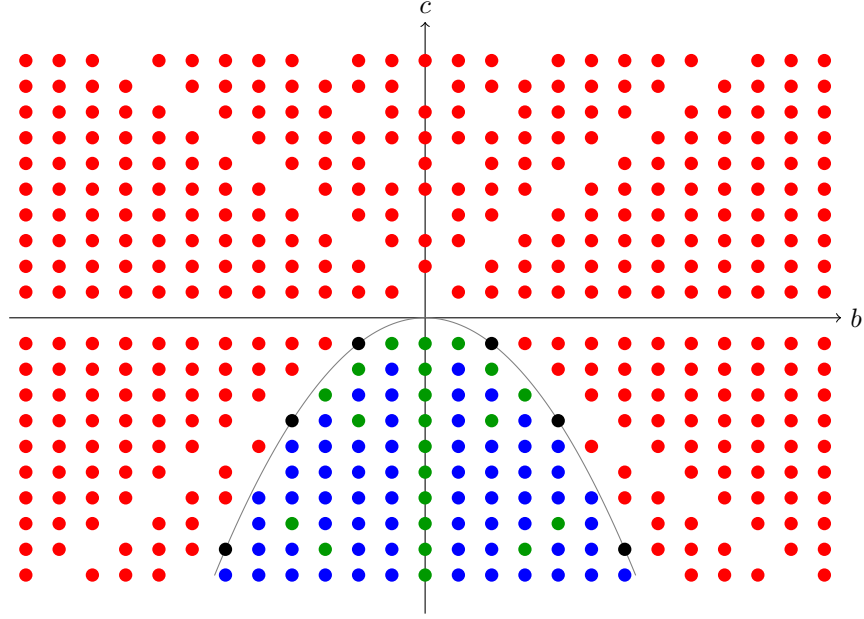

\subsection{New undecidability results for LRS}

Now fix an integer LRS $\seq{u_n}$ that
\begin{equation*}
    \textrm{has exactly two non-repeated dominant roots whose ratio is not a root of unity. \quad\: ($\maltese$)}
\end{equation*}
Note that $\seq{u_n}$ can have any order $d \ge 2$, as well as a reducible characteristic polynomial.
Under these assumptions, $u_n = a\lambda^n + \overline{a}\, \overline{\lambda}^n + r_n$ where $\lambda$ is non-real, $|\lambda| > 1$, and $|r_n|$ grows much slower than $|\lambda|^n$; see \Cref{sec::lrs-two-roots}.
The ratio of the two dominant roots $\lambda, \overline{\lambda}$ not being a root of unity is equivalent to both roots not being a root of a real number, and its role is to ensure non-degenerate behaviour (i.e., that $\seq{u_n}$ is not actually an interleaving of LRS with one real dominant root) as discussed above for order two LRS with $D < 0$.
Two concrete examples of integer LRS satisfying ($\maltese$) are
\begin{itemize}
    \item $u_{n+2} = u_{n+1}-2u_n$, $u_0 = 0$, $u_1 = 1$, $u_n = a \lambda^n + \overline{a} \, \overline{\lambda}^n$ with $a = 1/(\im \sqrt{7})$, $\lambda = (1 + \im \sqrt{7})/2$;
    \item $u_n = (2 + \im)^n + (2-\im)^n + 2^n$, which satisfies the recurrence relation $u_{n+3} = 6u_{n+2}-13u_{n+1} + 10u_n$.
\end{itemize}
Write $U = \{u_n \colon n \in \nat\}\cap\nat$ and $u \colon \nat \to \nat$, $u(n) = \max \{0,u_n\}$.
We will show in \Cref{sec::lrs-two-roots} that, assuming ($\maltese$), the predicate $U$ is always infinite. 
We order its elements as $\seq{p_n}$.
Our main result is the following.

\begin{theorem}
	\label{thm::undec-lrs}
	Suppose $\seq{u_n}$ satisfies $(\maltese)$.
    Then $\langle \nat; 0,1,<, u\rangle$ and $\langle \nat; 0,1,<,+, U \rangle$ simulate counter machines and therefore have undecidable first-order theories.
\end{theorem}

To us, the first result is again very surprising, since $\seq{u_n}$ can be very ``simple'', and $\langle \nat; 0,1,<, u\rangle$ has neither addition nor multiplication in its signature.\footnote{We do not know whether $+$ or $\cdot$ is definable in $\langle \nat; 0,1,<, u\rangle$; we suspect the answer is negative for both.}
We prove the two undecidability results of \Cref{thm::undec-lrs} respectively via the following number-theoretic results.

\begin{theorem}
	\label{thm::main-DA-function-version}
	Suppose $\seq{u_n}$ satisfies $(\maltese)$.
	Then there exists a computable $\zeta > 1$ with the following property.
	For all $\ell \ge 1$ and $1 < \gamma_j < \delta_j < \zeta$ for $1 \le j \le \ell$, there exist infinitely many $n < \widetilde{n}$ such that, writing $(n_j)_{j=0}^{k}$ for the ordering of 
    $
    \{m \colon n  \le m <\widetilde{n} \textrm{ and } u_n \le u_m < \zeta u_n\}
    $
    we have that $u_n > 0$, $k = \ell $, and
	\[
	\frac{u_{n_j}}{u_{n}} \in (\gamma_j, \delta_j)
	\]
	for all $1\le j \le \ell$.
\end{theorem}

\begin{theorem}
	\label{thm::main-DA}
	Suppose $\seq{u_n}$ satisfies $(\maltese)$.
	Then there exists a computable $\zeta > 1$ with the following property.
	For all $\ell \ge 1$ and $1 < \gamma_1 < \delta_1 \le \gamma_2 < \delta_2 \le \cdots \le \gamma_\ell < \delta_\ell < \zeta$, there exist infinitely many $n$ such that
	\[
	\frac{p_{n+j}}{p_n} \in (\gamma_j, \delta_j)
	\]
	for all $1 \le j \le \ell$.
\end{theorem}

These two results say that we can, to a certain extent, control the ratios of consecutive terms of $\seq{u_n}$ as well as the ordered version $\seq{p_n}$.
They are novel mathematical results of independent interest that are analogous to \Cref{thm::schinzel}.
However, because the randomness result we obtain for our LRS is weaker than Schinzel's result for the totient function, the process of extracting arbitrary finite sequences from LRS is a little more involved than just proving realisation of arbitrary permutations; see \Cref{sec::1st-structure,sec::2nd-structure}.


To prove \Cref{thm::main-DA-function-version,thm::main-DA}, we develop the theory of \emph{algebraic} sequences $a\lambda^n + \overline{a}\, \overline{\lambda}^n$ through novel Diophantine approximation arguments, to an extent that now LRS with two dominant roots are comparable to LRS with one dominant root in terms of how much we understand them.
Roughly speaking, we operate on the unit circle $\torus$ in $\com$, equipped with the function $z \mapsto z\gamma$ for a fixed $\gamma$ and a class of sequences $\seq{I^{(k)}_n}$ (where $k \in \mathcal{K}$ for an uncountable set~$\mathcal{K}$) of interval subsets of $\torus$ that disappear exponentially fast: for all $k\in\mathcal{K}$, $|I^{(k)}_n| \to 0$ as $n \to \infty$. 
There we argue about which intervals $I^{(k)}_n$ are hit and which are avoided by $\gamma^n \in \torus$ as $n \to \infty$.
However, we must mention that strong randomness properties of the class of LRS that we consider in this work were already hinted at in the recent paper~\cite{pronormality}, where the \emph{decidability} of the \emph{monadic second-order} (MSO) theory of $\langle \nat; <, U \rangle$ was established via the following theorem.\footnote{The paper \cite{pronormality} is about integer LRS $\seq{u_n}$ with two non-real dominant roots such that (i) all roots of $\seq{u_n}$ are non-repeated and (ii) the ratio of \emph{any} two characteristic roots of $\seq{u_n}$ is not a root of unity, which is stronger than ($\maltese$). However, the ratios and the repeatedness of non-dominant roots are irrelevant to their arguments, and all results of \cite{pronormality} that we cite apply under the assumption of ($\maltese$).}

\begin{theorem}[\unexpanded{\cite[Theorem~4]{pronormality}}]
	\label{thm::jojo}
	Let $\seq{u_n}$ be an integer LRS satisfying $(\maltese)$,  $m \ge 1$, and
	\[
	\Sigma_m = \{0 \le r < m \colon u_n \equiv r \, (\bmod\, m) \textrm{ for infinitely many $n$}\}.
	\]
	Then for any $t_0, \ldots, t_\ell \in \Sigma_m$ we can compute infinitely many $n$ such that for all $0 \le i \le \ell$,
	\[
	p_{n+i} \equiv t_i \, (\bmod\, m).
	\]
\end{theorem}

\section{Preliminaries}

We denote the cardinality of a set $X$ by $\# X$.
For $x \in \rel$ and $y \in \rel \setminus \{0\}$, let $[\![  x ]\!]_y$ be the smallest distance from $x$ to an integer multiple of~$y$.
We denote by $\im$ the imaginary number and by $\Log$ the principal branch of the complex logarithm, which satisfies $\Log(x+y\im) = \im \theta + \log \sqrt{x^2 + y^2}$, $\theta = \arg(x+y\im) \in (-\pi,\pi]$ for all $(x, y)\in \rel^2\setminus\{(0,0)\}$.
We write $\torus$ for the unit circle in $\com$, and $\torus_+$ for $\{z \in \torus \colon \Rea(z) > 0 \}$.
For $z_1,z_2 \in \torus$, we write $\Delta(z_1,z_2)$ for the length of a shortest arc of $\torus$ connecting $z_1$ and $z_2$.

\subsection{Structures and their theories}

A \emph{structure} $\mathbb{M}$ consists of a domain $D$, constants $c_1,\ldots,c_k \in D$, predicates $P_1,\ldots,P_l$ where each $P_i \subseteq D^{\mu(i)}$ for some $\mu(i) \ge 1$, and functions $f_1,\ldots,f_m$ where each $f_i$ has the type $f_i \colon D^{\delta(i)} \to D$ for some $\delta(i) \ge 1$.
We denote such $\Mb$ by $\langle D; c_1,\ldots,c_k, P_1,\ldots,P_l, f_1,\ldots,f_m\rangle$.
We do not explicitly mention $=$ as a predicate, but assume that every structure has it.
A \emph{theory} is simply a set of \emph{sentences}, i.e.\ first-order formulas without free variables.
The (first-order) theory of a structure~$\mathbb{M}$ is the set of all well-formed sentences constructed from the symbols $c_1,\ldots, c_k,P_1,\ldots,P_l, f_1,\ldots, f_m$ as well as $\land,\lor,\lnot, \exists,\forall$ that are true in $\mathbb{M}$.
We write $\Mb \models \varphi$ to mean that $\varphi$ holds in $\Mb$.
A formula is \emph{existential} if it is of the form $\exists x_1 \cdots \exists x_m \colon \varphi(x_1,\ldots,x_m)$ for $\varphi$ quantifier-free.
A theory $\Tcal$ is \emph{decidable} if there exists an algorithm that takes a sentence $\varphi$ and decides whether $\varphi \in \Tcal$ and \emph{undecidable} otherwise.

\subsection{Algebraic numbers}
An algebraic number $\alpha$ is a complex number that is a root of a nonzero polynomial $p \in \rat[x]$.
The unique monic polynomial $p \in \rat[x]$ of the smallest degree that has $\alpha$ as a root is called the \emph{minimal polynomial} of $\alpha$.
The set of algebraic numbers forms a field, written $\alg$.
An algebraic number $\alpha$ can be represented in computer memory, for example, by its minimal polynomial $p$ together with sufficiently close rational approximations to $\Rea(\alpha)$ and $\Ima(\alpha)$.
In this representation, all usual arithmetic operations can be effectively performed on algebraic numbers~\cite[Chapter~4]{cohen2013course}.

\subsection{Baker's theorem}

Baker's theorem and its $p$-adic analogue are among the most important mathematical tools in the study of linear recurrence sequences.
Let $\Lambda = b_1\Log(\alpha_1) + \cdots + b_m \Log(\alpha_m)$, where $b_i \in \intg$, $\alpha_i \in \alg \setminus \{0\}$.
Such a $\Lambda$ is called a linear form in logarithms.
The following is a quantitative version of Baker's theorem, see e.g.,~\cite{baker-rational-sharp-version-1993}.
\begin{theorem}[Baker's theorem]\label{thm:Baker}
    There exists a computable constant $C > 0$ (that only depends on $\alpha_1,\ldots,\alpha_m$) such that whenever $B \ge 3$, $B > |b_1|,\ldots,|b_m|$, and $\Lambda \ne 0$, 
	\[
	|\Lambda| > B^{-C}.
	\]
\end{theorem}

\begin{lemma}
	\label{lem::baker-distance-from-alpha-n-to-b}
	Let $\alpha, \beta \in \torus \cap \alg$. There exists a computable constant $C_0 > 0$ such that for any $n \in \nat$ satisfying $\alpha^n \ne \beta$, we have that	
	\[
	\Delta(\alpha^n, \beta) > \frac{1}{(\max\{2, n\})^{C_0}}.
	\]
\end{lemma}
\begin{proof}
	If $\alpha^n = -\beta$, $\Delta(\alpha^n, \beta) = \pi$ and any $C_0 > 0$ suffices.
    Else, we have that
	\begin{align*}
		\Delta(\alpha^n, \beta) &= [\![|n\Log(\alpha) - \Log(\beta)| ]\!]_{2\pi}\\
		&\ge  |n \Log(\alpha) - \Log(\beta) - k\pi\im | \\
		&= |n \Log(\alpha) - \Log(\beta) - k\Log(-1)| \ne 0
	\end{align*}
	where $-n-1 \le k \le n+1$.
	Applying Baker's theorem (\Cref{thm:Baker}), there is a computable constant $C > 0$ such that, setting $B = \max \, \{3, n+2\}$, 
	\[
	\Delta(\alpha^n, \beta) > \frac{1}{(\max\{3, n+2\})^{C}}.
	\]
	Finally, take $C_0 = 2C$. Then, $C_0 > 0$ and we claim that $(\max\{2,n\})^{C_0} \ge (\max\{3,n+2\})^{C}$. 
    Indeed, for $n \le 1$, $2^{2C} \ge 3^C$ as $C > 0$, and for $n \ge 2$, $n^{2C} \ge (n+2)^{C}$ as $2\log(n) \ge \log(n+2)$ for $n \ge 2$.
    The lemma follows.
\end{proof}

\subsection{Linear recurrence sequences}
\label{sec::lrs}

A sequence $\seq{u_n}$ over a ring $R$ is a \emph{linear recurrence sequence} (LRS) over $R$ if there exist $d \ge 0$ and $a_1,\dots,a_d \in R$, where $a_d \ne 0$ when $d \ge 1$, such that
\begin{equation}\label{eq::lrs-2}
	u_{n+d} = a_1u_{n+d-1}+ \cdots +a_du_n
\end{equation}
for all $n \in \nat$.
The smallest such $d$ is called the \emph{order} of $\seq{u_n}$.
In this paper, we work with LRS over $\intg$, which we also call \emph{integer LRS}.
For example, the Fibonacci sequence satisfies $u_{n+2} = u_{n+1} + u_n$ for all $n \in \nat$ and is an integer LRS of order two.
We refer the reader to the book~\cite{everest2003recurrence} for a detailed account of LRS\@.

Let $R \subseteq \alg$ be a ring and $\seq{u_n}$ be an LRS over $R$ of order~$d$.
Then there exist unique $a_1,\ldots,a_d \in R$ (with $a_d \ne 0$) such that $\seq{u_n}$ satisfies the recurrence relation \eqref{eq::lrs-2}.
The \emph{characteristic polynomial} of $\seq{u_n}$ is $p(x) = x^d - \sum_{i=1}^d a_i x^{d-i}$.
Suppose $p$ has the (distinct) roots $\lambda_1,\dots,\lambda_m \in \alg$, called the \emph{characteristic roots} of $\seq{u_n}$.
Then there exist unique non-zero polynomials $q_1,\ldots,q_m \in \alg[x]$ such that~\eqref{eq::lrs-1} holds for all $n \in \nat$ and $\deg(q_i)$ is at most the multiplicity of $\lambda_i$ as a root of the characteristic polynomial minus 1.
Equation~\eqref{eq::lrs-1} is known as the \emph{exponential-polynomial form} of $\seq{u_n}$.
A characteristic root $\lambda_i$ is called \emph{non-repeated} if $q_i$ is constant.
The sequence $\seq{u_n}$ is called \emph{diagonalisable} (alternatively, \emph{simple}) if every $\lambda_i$ is non-repeated.
A characteristic root $\lambda_i$ is called \emph{dominant} if $|\lambda_i| \ge |\lambda_j|$ for all $1 \le j \le m$.
We say that $\seq{u_n}$ is \emph{non-degenerate} if $z \coloneqq \lambda_i/\lambda_j$ is not a root of unity for all $i \ne j$, i.e., $z^k \ne 1$ for all $k \ne 0$.
For any LRS $\seq{u_n}$, there exists $L$ (that is effectively computable for integer LRS) such that the subsequences $\seq{u_{nL+r}}$ are non-degenerate for all $0 \le r < L$.
By the Skolem--Mahler--Lech theorem, every non-zero and non-degenerate LRS over $R$ (in fact, over any field of characteristic zero) has finitely many zeros.
We say that an integer LRS $\seq{u_n}$ is \emph{irreducible} if its characteristic polynomial is irreducible over~$\rat$.

Various fundamental decision problems of linear recurrence sequences are of central interest in computer science and mathematics~\cite{karimov2022decidable}.
The most famous example is the \emph{Skolem Problem}, which asks to decide whether a given integer LRS contains zero.
(An equivalent formulation is to decide whether a given integer linear \texttt{while} loop with a linear guard terminates~\cite{karimov2022linearloops}.)
At the time of writing, it is known to be decidable for non-degenerate LRS with at most 3 distinct dominant roots~\cite{tijdeman1984distance}, and remains open in general.

\subsection{LRS with two dominant roots}
\label{sec::lrs-two-roots}

In this section, fix an integer LRS $\seq{u_n}$ satisfying ($\maltese$).
If both dominant roots were real, then their ratio would be $-1$, which is a root of unity.
Hence, because the dominant roots are closed under complex conjugation, we have two non-real dominant roots $\lambda, \overline{\lambda}$.
We will adopt the following notation throughout this paper: $u_n = v_n + r_n$, $v_n = a\lambda^n + \overline{a}\, \overline{\lambda}^n$ where $a \in \alg \setminus \{0\}$, $\rho = |\lambda|  > 0$, $\lambda = \rho\mu$ with $\mu = e^{\im\theta} \in \torus$, and $a = |a|\xi$ with $\xi = e^{\im\varphi} \in \torus$.
We have that
\[
a\lambda^n + \overline{a}\, \overline{\lambda}^n = |a|\rho^n(\xi\mu^n + \overline{\xi}\,\overline{\mu}^n) = 2|a|\rho^n \cos(n\theta + \varphi)
\]
and, by the assumption on the dominant roots, $|r_n| = O((\rho-\varepsilon)^n)$ for all sufficiently small $\varepsilon > 0$ (where the implied constant is effective).
By non-degeneracy, $\theta$ is not a rational multiple of~$\pi$.
Finally, we mention that both $\seq{v_n}$ and $\seq{r_n}$ are themselves LRS over $\ralg$.

By Kronecker's theorem in Diophantine approximation~\cite{gonek2016kronecker}, because $\theta$ is not a rational multiple of $\pi$, we have that $\seq{\cos(n\theta + \varphi)}$ is dense in $[-1,1]$, which follows from $\seq{\xi \mu^n}$ being dense in $\torus$.
Then $\rho > 1$ due to \cite[Lemma 13]{pronormality}.

We next give a few lemmas for our class of LRS\@.
The following two results are proven using Baker's theorem.

\begin{lemma}[\unexpanded{\cite[Lemma 3]{vereshchagin1985occurrence}}]
	\label{thm::baker-applied-to-vn}
	There exist computable constants $N, C > 0$ such that for all $n \ge N$, $|v_n| > \frac{\rho^n}{n^C}$.
\end{lemma}

\begin{theorem}[\unexpanded{\cite[Theorem~3]{tijdeman1984distance}}]
	\label{thm::baker-application-distance-between-two-terms}
	There exist computable constants $N, C > 0$ such that for all $n \ge N$ and $m < n$,
	\[
	|u_n - u_m| > \rho^n n^{-C (\log (n+1))^2}.
	\]
\end{theorem}

In particular, by \Cref{thm::baker-applied-to-vn}, $|u_n| \to \infty$ as $n \to \infty$.
Because $\cos(n\theta+\varphi)$ is dense in $[-1,1]$, it follows that $u_n$ is infinitely often positive and infinitely often negative.
In particular, $U \coloneqq \{u_n \colon n \in\nat\} \cap \nat$ is infinite.

In the following lemmas, we show that, for our purposes, we can work with $v_n$ instead of $u_n$, provided that $n$ is sufficiently large.

\begin{lemma}
	\label{thm::from-un-to-vn}
	There exists computable $M$ with the following property.
	For all distinct $n_1, n_2 \in \nat$ with $n_1 \ge M$ we have that $u_{n_1} \ne u_{n_2}$, $u_{n_1} \ne 0$, $\operatorname{sign}(u_{n_1}) = \operatorname{sign}(v_{n_1})$, and $\operatorname{sign}(u_{n_1}-u_{n_2}) = \operatorname{sign}(v_{n_1}-v_{n_2})$.
\end{lemma}
\begin{proof}
	Let $N,C$ be as in \Cref{thm::baker-application-distance-between-two-terms}.
	Then for all $n_1 \ne n_2$ such that at least one $n_i \ge N$, we have that $u_{n_1} \ne u_{n_2}$.
	Since $|r_n| = O((\rho-\varepsilon)^n)$ for all sufficiently small $\varepsilon > 0$, applying \Cref{thm::baker-applied-to-vn} we can compute $N' \ge N$ such that $|v_n| > |r_n|$ and hence $\operatorname{sign}(u_n) = \operatorname{sign}(v_n)$ and $u_n \ne 0$ for all $n \ge N'$.
	Consider $n_1 \ge N'$ and $n_1 \ne n_2$.
	Exchanging $n_1$ and $n_2$ if necessary, we can assume that $n_1 > n_2$.
	Then
	\begin{equation}
		\label{eq::2-roots-reduction-to dominant}
		v_{n_1} - v_{n_2} = u_{n_1} - u_{n_2} + (r_{n_2} - r_{n_1}).
	\end{equation}
    Since $\seq{r_n} = O((\rho-\varepsilon)^n)$ for all sufficiently small $\varepsilon > 0$, we can compute $\widetilde{N}$ such that
	\[
	|r_{n_2} - r_{n_1}| < (\rho -\varepsilon)^{n_1}
	\]
	for all $n_1, n_2$ with $n_1 \ge \widetilde{N}$.
	Applying \Cref{thm::baker-application-distance-between-two-terms}, for all sufficiently large $n_1$ and (any) $n_2 < n_1$, we have that
	\[
	|u_{n_1}-u_{n_2}| > \rho^{n_1} n_1^{-C (\log (n_1+1))^2} > (\rho - \varepsilon)^{n_1}  > |r_{n_2}-r_{n_1}|
	\]
	which implies (together with \Cref{eq::2-roots-reduction-to dominant}) that
	\[
	\operatorname{sign}(v_{n_1} - v_{n_2}) = \operatorname{sign}(u_{n_1} - u_{n_2})  \in \{-, +\}.
	\qedhere
	\]
\end{proof}

\begin{lemma}
	\label{thm::vn-close-to-un}
	For every $\varepsilon \in \rat_{>0}$ there exists computable $M_\varepsilon$ such that for all $n \ge M_\varepsilon$
	\[
	1 - \varepsilon < \frac{|u_n|}{|v_n|} < 1+\varepsilon.
	\]
\end{lemma}
\begin{proof}
	We have
	\[
	\frac{u_n}{v_n} = 1 + \frac{r_n}{v_n}.
	\]
	It remains to observe that by \Cref{thm::baker-applied-to-vn},
	\[
	\lim_{n\to \infty} \frac{r_n}{v_n} = 0
	\]
	effectively.
\end{proof}

\subsection{The Hieronymi--Schulz interval stacking lemma}

The following lemma plays a crucial role in the proof of \Cref{thm::hs-randomness}, as well as our \Cref{thm::main-DA,thm::main-DA-function-version}.
\begin{lemma}[\unexpanded{\cite[Lemma~2.1]{hieronymi2022strong}}]
	\label{thm::hs}
	Let $I \subseteq \torus$ be an interval and $\seq{I_n}$ be a sequence of intervals such that $\sum_{n=0}^\infty |I_n \cap I| < |I|$ and $\seq{I_n}$ is dense in $I$.
	That is, for every non-empty open $J \subseteq I$ there exist infinitely many $n$ such that $I_n$ intersects $J$.
	Then there exist infinitely many $k$ such that $I_k \subseteq I$ and $I_k \cap I_n = \varnothing$ for all $n < k$.
\end{lemma}
\begin{proof}
	Let $N \in \nat$.
	We will construct $k > N$ with the required property.
	Let $J$ be a non-empty interval component of $I \setminus \bigcup_{n=0}^N I_n$ such that $|J| > \sum_{n=N+1}^\infty |I_n \cap J|$; such $J$ must exist by the assumption that $\sum_{n=0}^\infty |I_n \cap I| < |I|$.
	Let $z_1,z_2$ be two distinct points in $J$ outside $\bigcup_{n=0}^\infty I_n$.
	Further let $J' \subseteq J$ be an arc spanned by $z_1$ and $z_2$.
	We choose $k$ to be the smallest $n$ such that $I_n \cap J' \ne \varnothing$, which must exist by the density assumption.
	Such $I_k$ must also satisfy $I_k \subset J'$, as $z_1,z_2 \notin I_n$ for all $n$, and the conclusion follows.
\end{proof}

\section{Undecidability via simulation of counter machines}
\label{sec::how-to-prove-undec}

We defined the notion of a structure simulating counter machines in the Introduction.
We now formally prove that it implies undecidability of the attendant first-order theory.

\begin{proof}[Proof of \Cref{thm::how-to-prove-undec}]
    We will describe a procedure that takes as input a two-counter machine $\Mcal$, and outputs a formula $\Phi$ in the language of $\Mb$ that is true (in $\Mb$) if and only if $\Mcal$ halts.
	Recall that $\Mcal$ has counters $c_1,c_2$ (initialised to 1) that take positive integer values and instructions $1, \ldots, H$ for $H > 1$.
	The execution starts at line $l = 1$, and $l = H$ is the unique halting instruction.
	The trace of $\Mcal$ is $(0, \iota_0, c_{1,0}, c_{2,0} , 0, \iota_1, c_{1,1}, c_{2,1}, \ldots)$
	where $\iota_n, c_{1,n}, c_{2,n}$ are, respectively, the (next) instruction to be executed, the value of $c_1$, and the value of $c_2$ at time~$n$.
	We will construct a formula $\Phi$ stating that ``there exists $x$ such that $\mathsf{Seq}(x)$ is a finite sequence that is the trace of $\Mcal$ ending in the halting state'', which, by surjectivity of $\mathsf{Seq}$, is true if and only if $\Mcal$ halts.

	Define $\Phi \coloneqq \exists x \in D^l \colon
	\big(
	\Phi_{\mathsf{init}}(x) \land \Phi_{\mathsf{fin}}(x) \land \Phi_{\delta}(x)
	\big)$
	where
	\begin{align*}
		&\Phi_{\mathsf{init}}(x) \coloneqq \exists y_1,y_2,y_3,y_4 \in D^m \colon
		\bigg(
        \forall y \colon \lnot \frm{succ}(x, y, y_1) \:\land\: 
        \\
		&\quad \bigwedge_{i=1}^3 \frm{succ}(x,y_i,y_{i+1})
		\:\land\: \frm{cnst}_0(x,y_1) \:\land\: \bigwedge_{i=2}^4 \frm{cnst}_1(x,y_i) 
        \bigg)
        \\
		&\Phi_{\mathsf{fin}}(x) \coloneqq \exists y_1,y_2,y_3,y_4 \in D^m \colon \\
		&\quad \bigg(
		\bigwedge_{i=1}^3 \frm{succ}(x,y_i,y_{i+1})
		\:\land\: \frm{cnst}_0(x,y_1) \:\land\:
		\frm{cnst}_H(x, y_2)\:\land\:\forall z : \lnot \frm{succ}(x,y_4,z)
		\bigg)
		\\
		&\Phi_\delta(x) \coloneqq \forall y_1,\ldots,y_8 \in D^m \colon \\
		&\quad
		\bigg(
		\frm{cnst}_0(x,y_1) \:\land\: 
		\bigwedge_{i=1}^7 \frm{succ}(x, y_i, y_{i+1}) \Rightarrow \frm{cnst}_0(x,y_5) \:\land\: \Psi(x,y_2,y_3,y_4,y_6,y_7,y_8)
		\bigg)
	\end{align*}
	and $\Psi$ is a positive Boolean combination of the formulas implementing the transition function of $\Mcal$, i.e.\ a positive Boolean combination of formulas $\frm{inc}(x,y_i,y_j)$, $\frm{eq}(x,y_i,y_j)$, $\frm{cnst}_k(x,y_i)$, and $\lnot \frm{cnst}_k(x,y_i)$ for $i,j \in \{2,3,4,6,7,8\}$ and $1 \le k \le H$.
	The formula $\Phi_{\mathsf{init}}(x)$ states that the sequence $\mathsf{Seq}(x)$ starts with four consecutive terms that are $0, 1, 1, 1$, respectively.
	The formula $\Phi_{\mathsf{fin}}(x)$ states that $\mathsf{Seq}(x)$ contains four consecutive terms $s_1,\ldots,s_4$ such that $s_1 = 0$ and $s_2 = H$.
	Together, the two formulas imply that $\mathsf{Seq}(x)$ must contain at least eight terms: in particular, two blocks of four consecutive terms that start with the delimiter 0.
	Finally, $\Phi_\delta(x)$ states that whenever $s_1,\ldots,s_8$ are consecutive terms of $\mathsf{Seq}(x)$ such that $s_1 = 0$, then $s_5 = 0$ and the transition function of $\Mcal$ satisfies $\delta_{\Mcal}(s_2,s_3,s_4) = (s_6,s_7,s_8)$.    
\end{proof}

Recall that we also gave a function $f$ realising arbitrary permutations as a sufficient condition for undecidability of the first-order theory of $\langle\nat;0,1,<,f\rangle$.
We now formally prove this. 
We note that the construction for extracting arbitrary finite sequences over $\nat$ from an arbitrary permutation is the precursor to the more complex constructions for extracting arbitrary finite sequences via \Cref{thm::main-DA-function-version} or \Cref{thm::main-DA}.

\begin{proof}[Proof of \Cref{thm::permutation-to-undec}]
	For $c,d,e \in \nat$ we define
	\[
	\mathsf{Rep}(c,d,e) = ( d+1,d+2,\ldots,e)
	\]
	and $\mathsf{Seq}(c,d,e)$ to be the finite sequence $(t_i)_{i=1}^{e-d}$ over $\{0, \ldots, d-c\}$ where
	\[
	t_i =
	\#
	\big\{c < n \le d \colon f(n) < f(d+i)\big\}
	.
	\]
	Thus $d+i$ indexes the $i$th term of $\mathsf{Seq}(c,d,e)$.	
    Next, consider a non-empty sequence $(t_i)_{i=1}^N$ over $\{0,\ldots,R\}$; note that for the empty sequence, we can take, e.g., $c=d=e$.
	Let $\sigma$ be a permutation of $\{1,\ldots,N+R\}$ such that $\sigma^{-1}(1) < \cdots < \sigma^{-1}(R)$ and
	\[
	\# \big\{1 \le j \le R \colon \sigma^{-1}(j) < \sigma^{-1}(R+i)\big\} = t_i
	\]
	for all $1 \le i \le N$.
	Since $f$ realises arbitrary permutations, there exists $c \in \nat$ such that for all $c < n_1, n_2 \le c+N+R$,
	\[
	f(n_1) < f(n_2) \Leftrightarrow \sigma^{-1}(n_1-c) < \sigma^{-1}(n_2-c).
	\]
	We set $d = c+R$, $e = c+R+N$.
	Then, substituting $j = n-c$, we have for $1 \le i \le N$ that
	\begin{align*}
		\#\big\{c < n \le d \colon f(n) < f(d+i)\big\}
		&= \#\big\{c < n \le c+R \colon \sigma^{-1}(n-c) < \sigma^{-1}(R+i)\big\} \\
		&= \#\big\{1 \le j \le R \colon \sigma^{-1}(j) < \sigma^{-1}(R+i)\big\} = t_i.
	\end{align*}
	Therefore, $\mathsf{Seq}(c,d,e) =(t_i)_{i=1}^N$.
    
    Recall that we do not have $+$ in our signature; however, we can define the successor function $n \mapsto n+1$ via the formula 
    \[
    \psi(n, n') = n < n' \:\land\: \forall m: (m > n \Rightarrow m \ge n').
    \]
	Next, we define
	\begin{align*}
		&\frm{rep}(c,d,e,n) \coloneqq d+1 \le n \le e\\
		&\frm{cnst}_0(c,d,e,n) \coloneqq \frm{rep}(c,d,e,n) \:\land\: \lnot \exists m \in (c,d]\colon f(m) < f(n)\\
		&\frm{cnst}_k(c,d,e,n) \coloneqq  \frm{rep}(c,d,e,n) \:\land\: \exists! \{m_1, \ldots, m_k\} \subseteq (c,d] \colon  f(m_1), \ldots, f(m_k) < f(n) \\
		&\frm{succ}(c,d,e,n_1,n_2) \coloneqq \frm{rep}(c,d,e,n_1) \:\land\: \frm{rep}(c,d,e,n_2) \:\land\: n_2 = n_1 + 1\\
		&\frm{inc}(c,d,e,n_1,n_2) \coloneqq \frm{rep}(c,d,e,n_1) \:\land\: \frm{rep}(c,d,e,n_2) \:\land\: \\
		&\quad \exists ! m \in (c,d]\colon f(m) \in [f(n_1), f(n_2))\\
		&\frm{eq}(c,d,e,n_1,n_2) \coloneqq \frm{rep}(c,d,e,n_1) \:\land\: \frm{rep}(c,d,e,n_2) \:\land\: \\
		&\quad\lnot\exists m \in (c,d]\colon f(m) \in [f(n_1), f(n_2)) \cup [f(n_2), f(n_1))
	\end{align*}
	where $k \in \nat_{\ge1}$ and $\exists ! \{m_1, \ldots, m_k\} \subseteq (c,d]$ means ``there exists a unique set $\{m_1, \ldots,m_k\}$ of $k$ numbers from $(c,d]$''. (Similarly, $\exists ! m \in (c,d]$ means ``there exists unique $m \in (c,d]$''.)
	Observe that for valid representatives $n_1, n_2$ (i.e., when $d+1 \le n_1,n_2 \le e$),  $\frm{eq}(c,d,e,n_1,n_2)$ should evaluate to false if and only if there exists $m \in (c,d]$ such that either
	\begin{itemize}
        \item $f(m) < f(n_2)$ but $f(m) \nless f(n_1)$, or
		\item  $f(m) < f(n_1)$ but $f(m) \nless f(n_2)$.
	\end{itemize}
    
	That is, $ f(m) \in [f(n_1), f(n_2)) \cup [f(n_2), f(n_1))$.
	Similarly, for valid representatives $n_1, n_2$, $\frm{inc}(c,d,e,n_1,n_2)$ should evaluate to true if and only if $f(n_2) > f(n_1)$ and there exists unique $c < m \le d$ such that $f(m) < f(n_2)$ but $f(m) \nless f(n_1)$.
	The conditions (2-7) in the definition of simulating counter machines are thus satisfied, and we can apply \Cref{thm::how-to-prove-undec}.
\end{proof}

\section{The first-order theory of $\langle \nat; 0,1, <, u\rangle$}
\label{sec::1st-structure}

Let $\seq{u_n}$ be an integer LRS satisfying ($\maltese$) as in the Introduction.
In this section we prove that the first-order theory of $\Mb \coloneqq \langle \nat; 0,1, <, u\rangle$ is undecidable, which is the first half of \Cref{thm::undec-lrs}.
We do this by explicitly giving the maps $\mathsf{Seq} \colon \nat^3 \to \nat^*$, $\mathsf{Rep} \colon \nat^3 \to \nat^*$ and the formulas $\frm{rep}, \frm{cnst}_k,\frm{succ}, \frm{inc}, \frm{eq}$.
Our main tool for proving surjectivity of $\mathsf{Seq}$ is \Cref{thm::main-DA-function-version}, which will be proved in \Cref{sec::proof-of-DA-easy-version}.
Let $\zeta > 1$ be as in the statement of \Cref{thm::main-DA-function-version}.

Let $c,d,e \in \nat$, and suppose that $c< d < e$ and
\[
0 < u(c) < u(d) < u(e).
\]
(For $c,d,e$ that do not satisfy these conditions, we define $\mathsf{Seq}(c,d,e)$ and $\mathsf{Rep}(c,d,e)$ to be the empty sequence.)
Let
\[
X = \big\{c < n \le d \colon u(c) < u(n) < u(e)\big\}
\]
and $n_1 < \cdots < n_k$ be the ordering of all $d < n < e$ such that $u(c) < u(n) < u(e)$.
We extract  from $(c,d,e)$ a finite sequence $\mathsf{Seq}(c,d,e)$ of length $k$ by defining
\[
\mathsf{Rep}(c,d,e) \coloneqq  ( n_1, n_2, \ldots, n_k )
\]
and
\[
\mathsf{Seq}(c,d,e)_i \coloneqq \# \{x\in X \colon u(x) < u(n_i)\} \in \{0, \ldots, \#X\}
\]
for $1 \le i \le k$.
Note that this construction is slightly different   from the one we used in the previous section.

\begin{example}
	Consider the sequence defined by $u_{n+3} = -u_{n+2} + u_n$, $u_0 = 0$, $u_1 = 1$, and $u_2 = 2$.
	We have that
	\[
	u_n = a\lambda^n + \overline{a}\,\overline{\lambda}^n + y r^n
	\]
	where $\lambda \approx -0.88 + 0.74\im$, $r \approx 0.75$, $a \approx -0.58 + 0.61\im$ and $y \approx 1.17$.
	Let us compute $\mathsf{Seq}(c,d,e)$ and $\mathsf{Rep}(c,d,e)$ for $(c,d,e) = (81,92,100)$.
	(We mention that $u_{81}, u_{92}, u_{100} > 0$.)
	First, determine all $n \in (81,100)$ such that $u_{81} < u_n < u_{100}$.
	This leaves us with $n = 82$, $84$, $87$, $89$, $92$ , $94$, $95$, $99$.
	Since $d = 92$, $X = \{82,84,87,89,92\}$.
	Next, we observe that
	\label{ex1}
	\begin{align*}
		&{\color{blue} u_{81}} < u_{82} < u_{84} < u_{87} < u_{89} < {\color{red} u_{94}} < u_{92} < {\color{blue} u_{100}} \\
		&{\color{blue} u_{81}} < u_{82} < {\color{red} u_{95}} < u_{84} < u_{87} < u_{89} < u_{92} < {\color{blue} u_{100}} \\
		&{\color{blue} u_{81}} < u_{82} < u_{84} < u_{87} < {\color{red} u_{99}} < u_{89} < u_{92} < {\color{blue} u_{100}}.
	\end{align*}
	For each red term, we count the number of black terms that are smaller.
	(The blue terms are just delimiters, and are not counted.)
	Hence the triple $(c,d,e)$ defines the sequence $\mathsf{Seq}(81,92,100) = (4, 1, 3)$ with the corresponding representatives $\mathsf{Rep}(81,92,100) = (94, 95, 99)$.
\end{example}

We next show how the maps $\mathsf{Seq}$ and $\mathsf{Rep}$ can be implemented using first-order formulas.
We proceed similarly to the proof of \Cref{thm::permutation-to-undec} (\Cref{sec::how-to-prove-undec}) but need a more general successor function. 
Define
\begin{align*}
	&\frm{rep}(c,d,e,n) \coloneqq c < d < n < e \:\land\: 0 < u(c) < u(d) < u(e) \:\land\: u(c) < u(n) < u(e)\\
	&\frm{cnst}_0(c,d,e,n) \coloneqq \frm{rep}(c,d,e,n)
	\:\land\: \lnot \exists m \in (c,d] \colon u(m) \in (u(c),u(n))\\
	&\frm{cnst}_k(c,d,e,n) \coloneqq  \frm{rep}(c,d,e,n) \:\land\: \\
    &\quad\exists! \{m_1, \ldots, m_k\} \subset (c,d] \colon  u(m_1), \ldots, u(m_k) \in (u(c),u(n))\\
	&\frm{succ}(c,d,e,n_1,n_2) \coloneqq \frm{rep}(c,d,e,n_1) \:\land \: \frm{rep}(c,d,e,n_2) \:\land\: \\
    &\quad n_1 < n_2 \:\land\: \forall m \in (n_1,n_2) \colon \lnot\frm{rep}(c,d,e,m)\\
	&\frm{inc}(c,d,e,n_1,n_2) \coloneqq \frm{rep}(c,d,e,n_1) \:\land\: \frm{rep}(c,d,e,n_2) \:\land\: \\
	&\quad \exists ! m \in (c,d]\colon u(m) \in [u(n_1), u(n_2))\\
	&\frm{eq}(c,d,e,n_1,n_2) \coloneqq \frm{rep}(c,d,e,n_1) \:\land\: \frm{rep}(c,d,e,n_2) \:\land\: \\
	&\quad\lnot\exists m \in (c,d]\colon u(m) \in [u(n_1), u(n_2)) \cup [u(n_2), u(n_1)).
\end{align*}
where $k \in \nat_{\ge1}$, and $\exists! \{m_1, \ldots, m_k\}$ and $\exists !m$ are interpreted as in the proof of \Cref{thm::permutation-to-undec}.
Note that the definitions of $\frm{inc}$, and $\frm{eq}$ are of the same form (up to replacing $u(\cdot)$ with $f(\cdot)$) with those given in \Cref{sec::how-to-prove-undec}.
Items (2-7) in the definition of simulating counter machines are satisfied by construction.
To prove undecidability of the first-order theory of $\Mb$ (by applying \Cref{thm::how-to-prove-undec}) it remains to show the following.

\begin{lemma}
	\label{thm::seq-onto-easy-version}
	The map $\mathsf{Seq} \colon \nat^3 \to \nat^*$ above is surjective.
\end{lemma}
\begin{proof}
	To define the empty sequence, we can choose, for example, $(c,d,e) = (0,0,0)$.
	Now consider a non-empty finite sequence $(t_i)_{i=1}^N$ over $\{0, \ldots, R\}$ where, without loss of generality, we can enlarge the alphabet to ensure $R \ge 1$.
	We will construct, using \Cref{thm::main-DA-function-version},
	\[
	c = n_0 < n_1 < \cdots < n_R = d < n_{R+1} < \cdots < n_{R+N+1} = e
	\]
	such that
	\begin{itemize}
		\item[1)] $0 < u(c) < u(n_i) < u(e)$ for all $1\le i \le R + N$,
		\item[2)] $u(n) \notin (u(c), u(e))$ for any $n \in (c,e) \setminus \{n_1, \ldots, n_{R+N}\}$, and
		\item[3)] writing $X = \{n_1,\ldots,n_R\}$,
		\[
		\# \{x\in X \colon u(x) < u(n_{R+i})\}  = t_i
		\]
		for all $1 \le i \le N$.
	\end{itemize}
	Then $\mathsf{Seq}(c,d,e) = (t_i)_{i=1}^N$.

	\begin{example}
		\Cref{fig::1} illustrates our construction for the finite sequence $(1, 3, 0, 1)$.
		In this case $N = 4$ and $R = 3$.
		We partition $(1, \zeta)$ into $10$ intervals of equal length $\eta$.
		The intervals $(1, 1+\eta)$ and $(\zeta-\eta, \zeta)$ are just buffers: they are needed because in the statement of \Cref{thm::main-DA-function-version}, $\gamma_j$ and $\delta_j$ must be strictly between $1$ and~$\zeta$ for all $j$.
		We define
		\begin{itemize}
			\item $(\gamma_1, \delta_1) = (1+2\eta, 1+3\eta)$, $(\gamma_2,\delta_2) = (1+4\eta, 1+5\eta)$, $(\gamma_3,\delta_3) = (1+6\eta, 1+7\eta)$,
			\item $(\gamma_4, \delta_4) = (1+3\eta, 1+4\eta)$, $(\gamma_5,\delta_5) = (1+7\eta, 1+8\eta)$, $(\gamma_6,\delta_6) = (1+\eta, 1+2\eta)$, $(\gamma_7,\delta_7) = (1+3\eta, 1+4\eta)$, and
			\item $(\gamma_8, \delta_8) = (\zeta - 2\eta, \zeta-\eta)$.
		\end{itemize}
		Then $1 < \gamma_j < \delta_j < \zeta$ for all $j$.
		By \Cref{thm::main-DA-function-version} there exist infinitely many $n_0 < n_1 < \cdots < n_8$ such that $u_{n_j} > 0$ and,
		\[
		\frac{u(n_j)}{u(n_0)}  \in (\gamma_j, \delta_j)
		\]
		for all $j$, and for all integers $n \in [n_0, n_8] \setminus \{n_0, \ldots, n_8\}$, either $u(n) < u(n_j)$ for all $j$, or $u(n) > u(n_j)$ for all $j$.
		Note that because $\gamma_8 \ge \delta_j$ for $1 \le j \le 7$, we have that $u({n_0}) < u({n_j)} < u({n_8})$ for all $1 \le j \le 7$.
		From our construction of $(\gamma_j,\delta_j)$, $1 \le j \le \ell$ it then follows that  $\mathsf{Seq}(n_0, n_3, n_8) = (1, 3, 0, 1)$.
	\end{example}
	We now proceed with the proof.
	Let $\ell = R + N + 1$ and $\eta = (\zeta-1)/(2R+4)$.
	We define
	\[
	(\gamma_k,\delta_k) = (1 + 2k\eta, 1+(2k+1)\eta)
	\]
	for $1 \le k \le R$.
	For $k = 1,\ldots, N$, we define
	\[
	(\gamma_{R+k},\delta_{R+k}) = (1 + (2t_k+1)\eta, 1+(2t_k+2)\eta)
	\]
	Finally,  we set $(\gamma_\ell,\delta_\ell) = (1 + (2R+2)\eta, 1+(2R+3)\eta)$.
	Note that $1 < \gamma_j < \delta_j < \zeta$ for all $j$.
	Applying \Cref{thm::main-DA-function-version} we obtain infinitely many $n_0 < \cdots < n_\ell$ such that
	\begin{itemize}
		\item $u({n_j}) > 0$ for all $0 \le j \le \ell$,
		\item $\frac{u(n_j)}{u(n_0)} \in (\gamma_j, \delta_j)$ for all $1 \le j \le \ell$, which implies that $u({n_0}) < u({n_j}) < u({n_\ell})$ for such $j$, and
		\item $u(n) \notin (u(n_0), u(n_\ell))$ for all integers $n_0 < n  < n_\ell$ outside $\{n_1, \ldots, n_{\ell-1}\}$.
	\end{itemize}
	Thus $n_0,\ldots,n_\ell$ satisfy conditions (1-2) above.
	It remains to verify~(3).
	For all $1 \le i \le N$ we have that
	\begin{equation*}
		\#\{x\in X \colon u(x) < u(n_{R+i})\}
		=\#\bigg\{
		1 \le k \le R \colon \frac{u(n_k)}{u(n_0)} < \frac{u(n_{R+i})}{u(n_0)}
		\bigg\}.
	\end{equation*}
	Moreover,
	\[
	\frac{u(n_{R+i})}{u(n_0)} \in (1 + (2t_i+1)\eta, 1+2(t_i+1)\eta)
	\]
	and
	\[
	\frac{u(n_k)}{u(n_0)} \in (1 + 2k\eta, 1+(2k+1)\eta)
	\]
	for all $1\le k \le R$.
	Therefore,
	\[
	\bigg\{
	1 \le k \le R \colon \frac{u(n_k)}{u(n_0)} < \frac{u(n_{R+i})}{u(n_0)}
	\bigg\} = \{1, \ldots, t_i\}
	\]
	and hence $\#\{x\in X \colon u(x) < u(n_{R+i})\} = t_i$.
\end{proof}

\begin{figure}
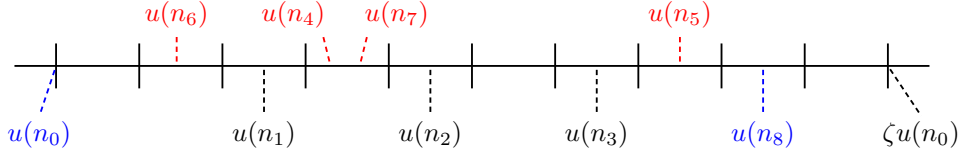

	\centering
	\seqlinefig{u(n_{#1})}
	\caption{The construction of the proof of \Cref{thm::seq-onto-easy-version} for the finite sequence $( 1, 3, 0, 1 )$. The horizontal line is $\rel$, and the vertical ticks are the points $u(n_0)(1 + \eta k)$ for $0 \le k \le 10$ and $\eta = (\zeta-1)/10$. The dashed lines indicate the locations of $u({n_0}), \ldots, u({n_8}), \zeta u({n_0})$.}
	\label{fig::1}
\end{figure}

\section{The first-order theory of $\langle \nat; 0, 1, <, +, U\rangle$}
\label{sec::2nd-structure}

Let $\seq{u_n}$ be an integer LRS satisfying ($\maltese$).
We now prove that $\Mb \coloneqq \langle \nat; 0, 1, <, +, U\rangle$ simulates counter machines, thus completing the proof of \Cref{thm::undec-lrs}.
Our main tool is \Cref{thm::main-DA}, which will be proven in \Cref{sec::proof-of-DA-hard-version}.

We first define the maps $\mathsf{Seq} \colon \nat^3 \to \nat^*$ and $\mathsf{Rep} \colon \nat^3 \to (\nat^2)^*$.
Recall that $<$ can be defined in $\langle \nat; +, U\rangle$ by $x < y \Leftrightarrow x \ne y \land \exists z \colon x + z = y$.
Consider $(p_c, p_d, p_e) \in U^3$ with $c < d < e$ (where $\seq{p_n}$ enumerates $U$ as in the Introduction);
for all other triples, both $\mathsf{Seq}$ and $\mathsf{Rep}$ return the empty sequence.
Let $R = d - c - 1$ and $N = e - d$.
For $1 \le i \le N$, we define
\[
\mathsf{Rep}(p_c,p_d,p_e)_i = (p_{d+i-1}, p_{d+i})
\]
and $\mathsf{Seq}(p_c,p_d,p_e)_i$ as
\[
\# \{c < n < d \colon p_n < p_c + p_{d+i} - p_{d+i-1}\} \in \{0, \ldots, R\}.
\]
We next show how to implement $\mathsf{Seq}$ and $\mathsf{Rep}$ using first-order formulas; these are modifications of the formulas given in \Cref{sec::1st-structure}.
We write $\mathbf{r}$ for the pair of variables $r_1,r_2$, and $\mathbf{\tilde{r}}$ for the pair $\tilde{r}_1, \tilde{r}_2$.
Define
\begin{align*}
	&\frm{rep}(y_1,y_2,y_3,\mathbf{r}) \coloneqq y_1,y_2,y_3 \in U \:\land\: y_1 < y_2 \le r_1 < r_2 \le y_3 \:\land\: r_1,r_2 \in U \:\land\: \\&\quad\forall r \in (r_1,r_2)\colon r \notin U \\
	&\frm{cnst}_0(y_1,y_2,y_3,\mathbf{r}) \coloneqq \frm{rep}(y_1,y_2,y_3,\mathbf{r}) \:\land\: \lnot \exists x \in (y_1,y_2) \cap U \colon x < y_1 + r_2 - r_1\\
	&\frm{cnst}_k(y_1,y_2,y_3,\mathbf{r}) \coloneqq \frm{rep}(y_1,y_2,y_3,\mathbf{r}) \:\land\: \\
	&\quad \exists! \{x_1,\ldots,x_k \}\subset (y_1,y_2) \cap U \colon x_1,\ldots,x_k < y_1 + r_2 - r_1\\
	&\frm{succ}(y_1,y_2,y_3,\mathbf{r},\mathbf{\tilde{r}}) \coloneqq \frm{rep}(y_1,y_2,y_3,\mathbf{r}) \: \land \frm{rep}(y_1,y_2,y_3,\mathbf{\tilde{r}}) \: \land \: r_2 = \tilde{r}_1
\end{align*}
where $k \in \nat_{\ge1}$ and $\exists! \{m_1, \ldots, m_k\}$ (as well as $\exists !$ used below) are interpreted as in \Cref{sec::how-to-prove-undec}.
It remains to define $\frm{eq}$ and $\frm{inc}$.
Write $t$ and $\tilde{t}$ for $y_1+r_2-r_1$ and $y_1+\tilde{r}_2-\tilde{r}_1$, respectively.
Then we can define
\begin{align*}
	&\frm{eq}(y_1,y_2,y_3,\mathbf{r},\mathbf{\tilde{r}}) \coloneqq \frm{rep}(y_1,y_2,y_3,\mathbf{r}) \:\land \: \frm{rep}(y_1,y_2,y_3,\mathbf{\tilde{r}}) \:\land\:  \\
	& \quad \forall x \in (y_1,y_2) \cap U \colon x \notin [t, \tilde{t}\,) \cup [\tilde{t}, t)\\
	&\frm{inc}(y_1,y_2,y_3,\mathbf{r},\mathbf{\tilde{r}}) \coloneqq \frm{rep}(y_1,y_2,y_3,\mathbf{r}) \:\land \: \frm{rep}(y_1,y_2,y_3,\mathbf{\tilde{r}}) \:\land\: \\
	& \quad \exists! x \in (y_1,y_2) \cap U \colon x \in [t, \tilde{t}\,).
\end{align*}
Items (2-7) in the definition of simulating counter machines are satisfied by construction.
It remains to prove the following.

\begin{lemma}
	The map $\mathsf{Seq} \colon \nat^3 \to \nat^*$ above is surjective.
\end{lemma}
\begin{proof}
	Take a non-empty finite sequence $(t_i)_{i=1}^N$ over $\{0,\ldots,R\}$.
	We will construct $c < d  < e$ such that $\mathsf{Seq}(p_c,p_d,p_e) = (t_i)_{i=1}^N$.
	(For the empty sequence, we can take any $c=d=e$.)
	It suffices to find $n$ such that
	\begin{equation}
		\label{eq::2nd-structure-surjective-1}
		\# \{n < m \le n+R \colon p_m < p_n + p_{n+R+i+1}-p_{n+R+i}\} = t_i
	\end{equation}
	for all $1 \le i \le N$.
	We can then choose $c=n$, $d = n+R+1$ and $e = n+R+N+1$.
	For $1 \le k \le R$, we define
	\[
	(\gamma_k, \delta_k) = (1+(2k-1)\eta, 1+2k\eta)
	\] where $\eta > 0$ is to be determined.
	Further define $(\gamma_{R+1}, \delta_{R+1}) = (1+(2R+1 - 1/8)\eta, 1+(2R+1 + 1/8)\eta)$.
	So far we have that
	\[
	\gamma_1 < \delta_1 < \gamma_2 < \delta_2 < \cdots < \gamma_R < \delta_R < \gamma_{R+1} < \delta_{R+1}.
	\]
	For $1 \le i \le N$, we write $k = R + i + 1$ and define
	\begin{align*}
		\gamma_k &= 1 + \bigg(-\frac 1 8 + 2R + 1 + \sum_{j=1}^i \big(2t_j + \frac 1 2\big)\bigg)\eta\\
		\delta_k &= 1 + \bigg(\frac 1 8 + 2R + 1 + \sum_{j=1}^i \big(2t_j + \frac 1 2\big)\bigg)\eta.
	\end{align*}
	Observe that $ \gamma_{R+1} < \delta_{R+1} < \cdots < \gamma_{R+1+N} < \delta_{R+1+N}$, and for any $z \in (\gamma_{R+i+1}, \delta_{R+i+1})$ and $x \in (\gamma_{R+i}, \delta_{R+i})$ we have that
	\[
	\frac{z-x}{\eta} \in \bigg(2t_i + \frac{1}{2} - \frac{1}{4}, 2t_i + \frac 1 2 + \frac 1 4\bigg),
	\]
	i.e., $z - x \approx (2t_i + 1/2)\eta$.
	Since
	\[
	1 < \gamma_j, \delta_j < 1 +2 \bigg(R + N + \sum_{i=1}^N t_i\bigg)\eta,
	\]
	for all $j$, we choose
	\[
	\eta = \frac{\zeta-1}{2 \big(R + N + \sum_{j=1}^N t_j\big)}.
	\]
	Then $1 < \gamma_j < \delta_j < \zeta$ for all $j$.

	Next, applying \Cref{thm::main-DA}, construct $n$ such that
	$
	\frac{p_{n+j}}{p_n} \in (\gamma_j, \delta_j)
	$
	for all $j$.
	It remains to show that \eqref{eq::2nd-structure-surjective-1} holds for all $1 \le i \le N$.
	Observe that
	\begin{multline*}
		\#\{n < m \le n+R \colon p_m < p_n + p_{n+R+i+1}-p_{n+R+i}\} \\
		=\#\bigg\{
		1 \le k \le R \colon \frac{p_{n+k}}{p_n} - 1 < \frac{ p_{n+R+i+1}-p_{n+R+i}}{p_n}
		\bigg\}.
	\end{multline*}
	Then it follows that
	$p_{n+R+i+1}/p_n \in (\gamma_{R+i+1}, \delta_{R+i+1})$ and
	$p_{n+R+i}/p_n \in (\gamma_{R+i}, \delta_{R+i})$.
	Therefore, as discussed earlier,
	\[
	\frac{p_{n+R+i+1}-p_{n+R+i}}{p_n} \in \bigg(\big(2t_i + \frac{1}{4}\big)\eta, \big(2t_i + \frac 3 4\big)\eta\bigg).
	\]
	On the other hand,
	\[
	\frac{p_{n+k}}{p_n} - 1 \in ((2k-1)\eta, 2k\eta)
	\]
	for all $1 \le k \le R$.
	Therefore,
	\begin{equation*}
		\bigg\{
		1 \le k \le R \colon \frac{p_{n+k}}{p_n} - 1 < \frac{ p_{n+R+i+1}-p_{n+R+i}}{p_n}
		\bigg\}
		= \{1, \ldots, t_i\}.
		\qedhere
	\end{equation*}
\end{proof}

\section{Proof of \Cref{thm::main-DA-function-version}}
\label{sec::proof-of-DA-easy-version}

Assume the notation of \Cref{sec::lrs-two-roots}.
Recall that we have
\begin{align*}
	u_n &= v_n + r_n  = a \lambda^n + \overline{a} \, \overline{\lambda}^n + r_n \\
	r_n &= o((\rho-\varepsilon)^n) \textrm{ for all sufficiently small $\varepsilon > 0$}\\
	v_n &= |a|\rho^n(\xi\mu^n + \overline{\xi}\,\overline{\mu}^n) =  2 |a| \rho^n \cos(n\theta + \varphi).
\end{align*}
Define
\begin{align*}
	\Ical &= \bigg\{z \in \torus \colon \Rea(z) > \frac{1}{\rho}\bigg\}.
\end{align*}
Then, by elementary geometry,
\begin{equation}
	\label{I-length}
	2 \cdot \sqrt{1 - \frac{1}{\rho^2}} <|\Ical| < \pi.
\end{equation}
The significance of $\Ical$ is as follows.
\begin{lemma}\label{lem::I-is-largest}
	For all $n$, if $\xi \mu^n \in \Ical$ then $v_n > 0, v_0,\ldots,v_{n-1}$.
\end{lemma}
\begin{proof}
	Suppose $\xi \mu^n \in \Ical$.
	Then $\xi \mu^n \in \torus_+$ and hence $v_n > 0$.
	Next, take $0 \le m < n$ such that $v_m > 0$; otherwise it is immediate that $v_n > v_m$.
	We have that
	\[
	\frac{v_n}{v_m} = \frac{\rho^{n} \Rea(\xi\mu^n)}{\rho^m \Rea(\xi\mu^m)} > \frac{\rho^{n-m}}{\Rea(\xi \mu^m)}\cdot \frac 1 \rho \ge \rho^{n-m-1} \ge 1.
	\qedhere
	\]
\end{proof}

For $z \in \torus_+$, let
\[
g_d(z) =  \frac{z \lambda^d + \overline{z} \, \overline{\lambda}^d}{z + \overline{z}}.
\]
Note that
\begin{equation}\label{eq::g_d(z)-square-expression}
	g_d(z) = \kappa \Leftrightarrow z^2 = -\frac{\overline{\lambda}^d - \kappa}{\lambda^d - \kappa}.
\end{equation}
For $\gamma,\delta \in \rel$ and $d \ge 1$, we define
\[
\Jcal_d(\gamma,\delta) = \{z \in \torus_+ \colon g_d(z) \in (\gamma,\delta)\}.
\]
We next argue that each $\Jcal_d(\gamma,\delta)$ is an interval.

\begin{lemma}
	For all $d \ge 1$, $g_d \colon \torus_+ \to \rel$ is a homeomorphism.
\end{lemma}
\begin{proof}
	Define $f_d \colon (-\pi/2,\pi/2) \to \rel$ by $f_d(x) = g_d(e^{\im x})$.
	We have that
	\[
	f_d(x) = \frac{\rho^d\cos(x+d\theta)}{\cos(x)}
	=
	\rho^d\big(\cos(d\theta)-\tan(x)\sin(d\theta)\big)
	\]
	where the last equality follows from the usual trigonometric relations.
    As $\mu = e^{\im\theta}$ is not a root of unity, $\sin(d\theta) \ne 0$.
	Because $\tan(x)$ is a homeomorphism from $(-\pi/2,\pi/2)$ to~$\rel$, so is $f_d$.
	It remains to write $g_d(z) = f_d(\Log(z)/\im)$, which is a composition of two homeomorphisms.
\end{proof}
\begin{corollary}
	\label{thm::Jd-are-open-intervals}
	For every $\gamma,\delta \in \rel$ with $\gamma < \delta$ and $d \ge 1$, $\Jcal_d(\gamma,\delta)$ is a non-empty and open interval.
\end{corollary}
\begin{proof}
	We have that
	\[
	\Jcal_d(\gamma,\delta) = \{z \in \torus_+ \colon g_d(z) \in (\gamma,\delta)\} = g_d^{-1}((\gamma,\delta)). \qedhere
	\]
\end{proof}

The  intervals $\Jcal_d(\gamma,\delta)$ play a key role in the proof of \Cref{thm::jojo}~\cite{pronormality}.
The idea behind their definition is that for all $n \in \nat$ and $d \ge 1$,
\begin{align}
	\label{eq::fallin-into_Jd}
	\xi\mu^n \in \Jcal_d(\gamma,\delta) &\Leftrightarrow  v_n > 0 \:\land\: \frac{v_{n+d}}{v_n} \in (\gamma, \delta).
\end{align}
Recall from \Cref{sec::lrs-two-roots} that $\seq{\xi \mu^n}$ is dense in $\torus$.
Hence \Cref{thm::Jd-are-open-intervals} tells us that for any $d \ge 1$ and $\gamma < \delta$, we can find infinitely many $n$ such that $v_n > 0$ and
\[
\frac{v_{n+d} }{v_n} \in (\gamma,\delta).
\]
Proving \Cref{thm::main-DA-function-version} amounts to proving a version of the preceding statement that involves arbitrarily many terms of $\seq{v_n}$ as opposed to only $v_n$ and $v_{n+d}$.

We next estimate the length of $\Jcal_d( \gamma,\delta)$.
\Cref{thm::int-sizes-upper} is similar to a result proven in~\cite{pronormality}, whereas \Cref{thm::int-sizes-lower} is much stronger than the analogous result from~\cite{pronormality}.
Recall
that we denote by $\Delta$ the arc distance function on the unit circle $\torus \subseteq \com$.

\begin{lemma}
	\label{thm::int-sizes-upper}
	There exists computable $C_1 > 0$ with the following property.
	For any $0 \le \gamma < \delta \le \frac{\rho+1}{2}$ and $d \ge 1$,
	\[
	|\Jcal_d(\gamma,\delta)| < \frac{C_1(\delta-\gamma)}{\rho^d}.
	\]
\end{lemma}
\begin{proof}
	Let $z_1,z_2 \in \torus_+$ be such that $g_d(z_1) = \gamma$  and $g_d(z_2)~=~\delta$.
	Then $|\Jcal_d(\gamma,\delta)| = \Delta(z_1,z_2)$.
	By \eqref{eq::g_d(z)-square-expression},
	\begin{align*}
		z_1^2 &= -\frac{\overline{\lambda}^d - \gamma}{\lambda^d - \gamma}\\
		z_2^2 &= -\frac{\overline{\lambda}^d - \delta}{\lambda^d - \delta}.
	\end{align*}
	By the geometry of the unit circle, we have that
	\begin{equation}
		\label{eq::int-size-lower-1}
		|z_1 - z_2| < \Delta(z_1,z_2) < \frac{\pi}{2}|z_1-z_2|.
	\end{equation}
	Hence it suffices to estimate
	\begin{equation}
		\label{eq::int-size-lower-2}
		|z_1 - z_2| = \frac{|z_1^2-z_2^2|}{|z_1 + z_2|}.
	\end{equation}
	We first argue that $\Delta(z_1,z_2) < \pi/2$, which implies that
	\begin{equation}
		\label{eq::int-size-lower-3}
		\sqrt{2} < |z_1 + z_2| < 2.
	\end{equation}
	Suppose that $\Delta(z_1,z_2) \ge~\pi/2$.
	Then, by the continuity of $g_d$, there exists $\gamma < \beta \le \delta$ and $z_3 \in \torus_+$ such that $g_d(z_3) = \beta$ and $\Delta(z_1,z_3) = \pi/2$.
	Then either $z_3 = \im z_1$ or $z_3 = - \im z_1$, and hence $z_1^2 = -z_3^2$.
	Applying \Cref{eq::g_d(z)-square-expression}, we obtain that
	\[
	\frac{\overline{\lambda}^d - \gamma}{\lambda^d - \gamma}=-\frac{\overline{\lambda}^d - \beta}{\lambda^d - \beta}
	\]
	which simplifies to
	\[
	2\lambda^d \overline{\lambda}^d - (\beta+\gamma)(\overline{\lambda}^d+\lambda^d) + 2\beta\gamma = 0.
	\]
	Observe that $\lambda^d \overline{\lambda}^d = \rho^{2d}$, $\lambda^d+\overline{\lambda}^d$, and $0 \le \gamma < \beta \le (\rho+1)/2$ are all real and that $\lambda^d+\overline{\lambda}^d \le 2\rho^d$. 
	Hence, by factoring, $(\rho^d-\beta)(\rho^d-\gamma)\le 0$, which implies that $\rho^d \le \beta$. However, $\beta\le(\rho+1)/2 < \rho$, contradicting that $d\ge 1$.
	We conclude that $\Delta(z_1,z_2)  < \pi/2$.
	
	Therefore, to prove an upper bound on $\Delta(z_1,z_2)$ it suffices to consider $|z_1^2-z_2^2|$.
	We have
	\begin{equation}
		\label{eq::int-size-lower-4}
		|z_1^2-z_2^2| = \frac{(\delta-\gamma)\rho^d}{|\lambda^d-\gamma||\lambda^d-\delta|}|\mu^{2d}-1|.
	\end{equation}
	Consider $|\lambda^d - \gamma|$.
	By assumption, $|\lambda^d| = \rho^d > \gamma$.
	Moreover, $\lambda^d$ is non-real.
	Hence $|\lambda^d - \gamma| > \rho^d - \gamma \ge \rho^d - \frac{\rho+1}{2} = \rho^d (1- \frac{\rho+1}{2\rho^d}) \ge  \rho^d \cdot c$ where
	$c = 1- \frac{\rho+1}{2\rho} > 0$.
	Similarly, $|\lambda^d - \delta| > \rho^d \cdot c$.
	Therefore, we only need to bound $|\mu^{2d}-1|$ from above by a constant.
	Since $|\mu^{2d}-1| \le 2$, we obtain
	\[
	|z_1^2-z_2^2| < \frac{2(\delta-\gamma)}{c^2 \rho^d}.
	\]
	and hence
	\[
	\Delta(z_1,z_2) < \frac{\pi}{2} \cdot \frac{2(\delta-\gamma)}{c^2 \rho^d} \cdot \frac{1}{\sqrt{2}}.\qedhere
	\]
\end{proof}
For $d \ge 1$ let $\alpha_d$ be the unique $z \in \torus_+$ such that $g_d(z) = 0$.
For any $d$, $\gamma$, and $\delta$, we refer to $\alpha_d$ as the \emph{anchor point} of $\Jcal_d(\gamma,\delta)$.
The reason for this is that for any fixed $\gamma < \delta$, as $d \to \infty$, the intervals $\Jcal_d(\gamma,\delta)$ become arbitrarily small while getting arbitrarily close to the point $\alpha_d$.
Note that $\alpha_d$ satisfies $\alpha_d \lambda^d + \overline{\alpha_d} \overline{\lambda}^d = 0$ and $|\alpha_d| = 1$.
Hence $\alpha_d \in \{\im \mu^{-d}, -\im \mu^{-d}\}$.

\begin{lemma}
	\label{thm::int-sizes-lower}
	There exists computable $C_2 > 0$ with the following property.
	Let $0 \le \gamma < \delta \le \frac{\rho+1}{2}$, $d \ge 1$, and suppose $\Jcal_d(\gamma,\delta) \subseteq \Ical$.
	Then
	\[
	|\Jcal_d(\gamma, \delta)| > \frac{C_2(\delta-\gamma)}{\rho^d}.
	\]
\end{lemma}
\begin{proof}
	Define $z_1$ and $z_2$ as was done in the proof of \Cref{thm::int-sizes-upper}; then $|\Jcal_d(\gamma,\delta)| > |z_1 - z_2|$.
	From \Cref{eq::int-size-lower-1,eq::int-size-lower-2,eq::int-size-lower-3,eq::int-size-lower-4}
	we obtain
	\[
	|z_1-z_2| > \frac{(\delta-\gamma)\rho^d}{2|\lambda^d-\gamma||\lambda^d-\delta|}|\mu^{2d}-1|.
	\]
	By the triangle inequality, $|\lambda^d-\gamma|, |\lambda^d-\delta| < (1 + \frac{\rho+1}{2})\rho^d$.
	Hence it suffices to give a constant lower bound for $|\mu^{2d}-1|$, which we do below.
	
	Let $\widetilde{\Ical} = \{z \in \torus_+ \colon \Rea(z) > 1/(2\rho)\} \supset \Ical$ and $D \ge 1$ be such that $|\Jcal_d(0, \frac{\rho+1}{2})| < \frac{1}{2\rho}$ for all $d \ge D$.
	Then for all $d \ge D$, if $\Jcal_d(\gamma,\delta) \subseteq \Ical$
	then $\Jcal_d(0, \frac{\rho+1}{2}) \subseteq \widetilde{\Ical}$ and, as $\alpha_d$ is an endpoint of $\Jcal_d(0, \frac{\rho+1}{2})$, $\alpha_d \in \widetilde{\Ical}$.
	We define
	\[
	A = \min_{1\le d \le D} |\mu^{2d}-1|
	\]
	which is positive by the assumption that $\mu$ is not a root of unity.
	We will show that
	\[
	|\mu^{2d}-1| \ge \min \{A, 1/(4\rho^2)\}.
	\]
	If $d \le D$, this is immediate.
	Now suppose $d \ge D$.
	Recall that $\mu^d = \frac{\im}{\pm \alpha_d}$ where $\pm \alpha_d$ is one of $\alpha_d$ and $-\alpha_d$.
	Then, as $\alpha_d \in \widetilde{\Ical}$, we have that $\Rea(\alpha_d) > 1/(2\rho)$ and hence
	\[
	|\mu^{2d}-1| = |\mu^d-1| \cdot |\mu^d+1| = |\pm\alpha_d - \im|\cdot |\pm\alpha_d+\im| \ge 1/(4\rho^2).\qedhere
	\]
\end{proof}

We will also need the following lemma, which is analogous to the density of $\xi \mu^n \in \torus$ proven via Kronecker's theorem.

\begin{lemma}
	\label{thm::density-of-intervals}
	For any $0 \le \gamma < \delta < (\rho+1)/2$, the intervals $(\Jcal_d(\gamma,\delta))_{d=1}^\infty$ are dense in $\torus_+$.
\end{lemma}
\begin{proof}
	Recall that $\alpha_d$ is an endpoint of $\Jcal_d(0, \delta) \subseteq \torus_+$, and that $\lim_{d\to \infty} |\Jcal_d(0,\delta)| = 0$ by \Cref{thm::int-sizes-upper}.
	Therefore, it suffices to prove that $(\alpha_d)_{d=1}^\infty$ is dense in $\torus_+$.
	Recall that $\alpha_d$ is either $\im\mu^{-d}$ or $-\im \mu^{-d}$, whichever is in $\torus_+$. (Exactly one of these is always the case, since by the non-degeneracy assumption, $\mu$ is not a root of unity.)
	By Kronecker's theorem, $(\im\mu^{-d})_{d=1}^\infty$ is dense in $\torus$.
	Since $(\alpha_d)_{d=1}^\infty$ contains all terms of  $(\im\mu^{-d})_{d=1}^\infty$ that lie in $\torus_+$, it is dense in~$\torus_+$.
\end{proof}

Henceforth fix $C_1,C_2 > 0$ as in the two lemmas above.
Choose $\zeta$ such that
\begin{equation}
	\label{zeta-def}
	\zeta > 1, \qquad
	\zeta <  \frac{\rho+1}{2}, \qquad
	\frac{C_1(\zeta-1)}{\rho-1} < |\Ical|.
\end{equation}
We will only work with subintervals of $\Jcal_d(1,\zeta)$ for $d \ge 1$.

The next lemma tells us that as $d_1 \to \infty$, it becomes exponentially harder for $\Jcal_{d_2}(1,\zeta)$ with $d_2 > d_1$ to intersect $\Jcal_{d_1}(1,\zeta)$.

\begin{lemma}
	\label{thm::two-interval-int}
	Suppose $d_1  < d_2$ and
	\begin{equation}
		\label{eq::two-interval-int-1}
		\Jcal_{d_1}(1,\zeta) \cap \Jcal_{d_2}(1,\zeta) \ne \varnothing.
	\end{equation}
	Then
	\begin{equation}
		d_2-d_1 > \bigg(\frac{\rho^{d_1}}{2C_1\zeta}\bigg)^{1/C_0} - 1
	\end{equation}
	where $C_0, C_1$ are the constants of \Cref{lem::baker-distance-from-alpha-n-to-b,thm::int-sizes-upper}, respectively.
\end{lemma}
\begin{proof}
		By the triangle inequality,
		\begin{equation}\label{eq::triangle-two-intervals}
			\Delta(\alpha_{d_1}, \alpha_{d_2}) \le \Delta(\alpha_{d_1}, z) + \Delta(z, \alpha_{d_2})
		\end{equation}
		for $z \in \Jcal_{d_1}(1,\zeta) \cap \Jcal_{d_2}(1,\zeta)$.
		Recall that $\alpha_d$ is an endpoint of $\Jcal_d(0, \zeta)$.
		Hence we have
		\begin{equation}
			\Delta(\alpha_{d_1}, \alpha_{d_2}) \le |\Jcal_{d_1}(0,\zeta)| + |\Jcal_{d_2}(0,\zeta)| < \frac{2C_1\zeta}{\rho^{d_1}}
		\end{equation}
		where the last inequality is deduced from \Cref{thm::int-sizes-upper}.
		Write $\alpha_{d_i} = \chi_i\mu^{-d_i}$ for $i = 1,2$ where $\chi_i \in \{\im, -\im\}$.
		Then
		\begin{align*}
			\Delta(\alpha_{d_1}, \alpha_{d_2}) = \Delta(\chi_1\chi_2^{-1}, \mu^{-d_2+d_1}).
		\end{align*}
		Note that $\alpha_{d_1} \ne \alpha_{d_2}$ as otherwise we would have $\mu^{-d_1} = \mu^{-d_2}$ or $\mu^{-d_1} = - \mu^{-d_2}$, contradicting the assumption that $\mu$ is not a root of unity.
		Applying \Cref{lem::baker-distance-from-alpha-n-to-b} on the left-hand side (note that $\max \, \{2,d_2-d_1\} \le 1 + d_2-d_1$),
		\[
		\Delta(\alpha_{d_1}, \alpha_{d_2}) > \frac{1}{(1 + d_2 - d_1)^{C_0}}.
		\]
		Therefore,
		\[
		\frac{1}{(1 + d_2 - d_1)^{C_0}} <  \frac{2C_1\zeta}{\rho^{d_1}}.
		\]
		Rearranging gives the desired conclusion.
\end{proof}

Next, we further study when $\Jcal_{d_1}(1,\zeta) \cap \Jcal_{d_2}(1,\zeta) \ne \varnothing$.

\begin{definition}
	An interval $I \subseteq \Ical$ is protected from time $d \ge 0$ onwards if
	\[
	|I| > \sum_{k=d+1}^\infty |I \cap \Jcal_k(1,\zeta)|.
	\]
\end{definition}

\begin{lemma}
	\label{thm::I-protected}
	The interval $\Ical$ is protected from time $0$ onwards.
\end{lemma}
\begin{proof}
	By \Cref{thm::int-sizes-upper} and the construction of $\zeta$,
	\[
	\sum_{d=1}^\infty |\Ical \cap \Jcal_d(1,\zeta)| \le \sum_{d=1}^\infty |\Jcal_d(1,\zeta)| \le \frac{C_1(\zeta-1)}{\rho-1} < |\Ical|.
	\qedhere
	\]
\end{proof}

We are now ready to prove a version of \Cref{thm::main-DA-function-version} for $\seq{v_n}$; soon thereafter we will move to $\seq{u_n}$.
First, a helpful lemma.

\begin{lemma}
	\label{thm::helpful-1}
	Let $1 < \gamma < \delta < \zeta$.
	There exists $D \ge 0$ such that for all $d \ge D$, if $\Jcal_d(\gamma, \delta) \subseteq \Ical$ then $\Jcal_d(\gamma, \delta)$ is protected from time $d$ onwards.
\end{lemma}
\begin{proof}
	Applying \Cref{thm::two-interval-int}, we have that for all $d > 0$,
		\begin{align*}
			\sum_{k=d+1}^\infty |\Jcal_{d}(\gamma,\delta) \cap \Jcal_k(1,\zeta)|
			\le
			\sum_{k=\chi(d)}^\infty | \Jcal_k(1,\zeta)| 
			\le \frac{C_1 (\zeta-1)}{\rho^{\chi(d)-1} \cdot (\rho-1)}
		\end{align*}
		where
		\[
		\chi(d) = \bigg\lceil d + \bigg(\frac{\rho^{d}}{2C_1\zeta}\bigg)^{1/C_0}-1  \bigg\rceil.
		\]
		Because $\chi(d)$ grows much faster than $d$, we can compute $D$ such that for all $d \ge D$,
		\[
		\frac{C_1 (\zeta-1)}{\rho^{\chi(d)-1} \cdot (\rho-1)} < \frac{C_2(\delta-\gamma)}{\rho^d}.
		\]
		Therefore, for all $d \ge D$, if $\Jcal_d(\gamma,\delta) \subseteq \Ical$ then by \Cref{thm::int-sizes-lower},
		\[
		|\Jcal_d(\gamma,\delta)| > \frac{C_2(\delta-\gamma)}{\rho^d} > \sum_{k=d+1}^\infty |\Jcal_{d}(\gamma,\delta) \cap \Jcal_k(1,\zeta)|.
		\qedhere
		\]
\end{proof}

\begin{lemma}
	\label{thm::stacking-intervals}
	Suppose we are given $\ell \ge 1$ and for $1 \le j \le \ell$, $1 < \gamma_j < \delta_j < \zeta$.
	We can construct $0 < d_1 < \cdots < d_\ell$ with the following properties.
	\begin{itemize}
		\item[(a)] $\Ical \supseteq \Jcal_{d_1}(\gamma_1,\delta_1) \supseteq \cdots \supseteq \Jcal_{d_\ell}(\gamma_\ell,\delta_\ell)$.
		\item[(b)] For all $1 \le d \le d_\ell$, if $d \notin \{d_1,\ldots,d_\ell\}$ then
		\[
		\Jcal_d(1,\zeta) \cap \Jcal_{d_\ell}(1, \zeta) = \varnothing.
		\]
		\item[(c)] Every $\Jcal_{d_j}(\gamma_j, \delta_j)$ is protected from time $d_j$ onwards.
	\end{itemize}
\end{lemma}
\begin{proof}
	We proceed by induction on $\ell$.
    For $\ell = 1$, by \Cref{thm::I-protected,thm::density-of-intervals}, we can apply \Cref{thm::hs} to $\Ical$ and $(\Jcal_d(1,\zeta))_{d\ge1}$ and let $d_1$ be sufficiently large. 
    Then (a) and (b) follow from $\Jcal_{d_1}(\gamma_1,\delta_1) \subseteq \Jcal_{d_1}(1,\zeta) \subseteq \Ical$, and (c) from \Cref{thm::helpful-1}.

	Next, consider $\ell = m + 1 \ge 2$.
	Apply the induction hypothesis with $\gamma_1, \delta_1, \ldots, \gamma_m,\delta_m$ to construct $d_1, \ldots, d_m$.
	By \Cref{thm::density-of-intervals}, $(\Jcal_d(1,\zeta))_{d > d_m}$ are dense in $\Jcal_{d_m}(\gamma_m, \delta_m)$.
	Applying~(c) of the induction hypothesis and \Cref{thm::hs}, there exist infinitely many $d > d_m$ such that $\Jcal_d(1,\zeta) \subseteq \Jcal_{d_m}(\gamma_m, \delta_m) \subseteq \Ical$ and for all $1 \le d' < d$ not equal to any $d_i$, $\Jcal_d(1,\zeta) \cap \Jcal_{d'}(1,\zeta) = \varnothing$.
	Applying \Cref{thm::helpful-1}, we can construct infinitely many $d$ for which we additionally have that $\Jcal_d(\gamma_\ell, \delta_\ell)$ is protected from time $d$ onwards.
	We can then pick $d_\ell$ to be equal to any such~$d$.
\end{proof}

Before finally proving \Cref{thm::main-DA-function-version}, let us have a look at what \Cref{thm::stacking-intervals} immediately gives us.
Consider $1 < \gamma_j < \delta_j < \zeta$ for $1 \le j \le \ell$.
Construct $d_1,\ldots,d_\ell > 0$ using \Cref{thm::stacking-intervals}.
By the density of $\seq{\xi\mu^n}$ in $\torus$, there exist infinitely many $\xi \mu^n \in \Jcal_{d_\ell}(\gamma_\ell, \delta_\ell) \subseteq \Ical$.
Then $v_n > 0$, and by~(a), $\xi \mu^n \in \Jcal_{d_j}(\gamma_j, \delta_j)$ for all $j$.
Write $n_j = n + d_j$ for $1 \le j \le \ell$.
By \Cref{eq::fallin-into_Jd}, 
\[
\frac{v_{n_j}}{v_n} \in (\gamma_j, \delta_j) \subset (1, \zeta)
\]
for all $1 \le j \le \ell$.
On the other hand, from~(b) it follows that for all $n \le m \le n_\ell$,
\begin{equation}
	\label{eq::type2-main-3}
	m \notin \{n_1,\ldots,n_\ell\} \Rightarrow \frac{v_m}{v_n} \notin (1, \zeta).
\end{equation}
Therefore, the pattern $v_{n}, v_{n_1}, \ldots, v_{n_\ell}$ appears in the ordering of $(v_m)_{m=n}^{n_\ell}$, and has the ratios controlled by $(\gamma_j,\delta_j)$ for $1 \le j \le \ell$.

\begin{proof}[Proof of \Cref{thm::main-DA-function-version}]
	Let $\zeta$ be as in \Cref{zeta-def} and $M$ be as in \Cref{thm::from-un-to-vn}.
	Suppose we are given $1 < \gamma_j < \delta_j < \zeta$ for $1 \le j \le \ell$.
	First construct rationals $\widetilde{\gamma}_j, \widetilde{\delta}_j$ satisfying $\gamma_j < \widetilde{\gamma}_j < \widetilde{\delta}_j < \delta_j$ for $1 \le j \le \ell$.
	Apply \Cref{thm::stacking-intervals} with $\widetilde{\gamma}_j, \widetilde{\delta}_j$ to construct $d_1 < \cdots < d_\ell$, and let $I = \Jcal_{d_\ell}(\widetilde{\gamma}_\ell,\widetilde{\delta}_\ell)$ and $D = \{1 \le d \le d_\ell \colon d \ne d_1,\dots,d_\ell\}$.
	As $g_{d_\ell}$ is a homeomorphism, $\overline{I} = g_{d_\ell}^{-1}([\widetilde{\gamma}_\ell,\widetilde{\delta}_\ell])$ is a compact arc contained in the open set $\Jcal_{d_\ell}(1,\zeta)$.
	By \Cref{thm::stacking-intervals}~(b), for $d \in D$ the latter is disjoint from $\Jcal_d(1,\zeta)$.
	Thus $g_d(\overline{I})$ is a compact interval disjoint from $[1,\zeta]$, and either $\max_{\overline{I}} g_d < 1$ or $\min_{\overline{I}} g_d > \zeta$.
	Set
	\begin{align*}
		\tilde{\zeta}_- &= \max\left(\{1/2\} \cup \left\{ \max_{z \in \overline{I}} g_d(z) \;\colon\; d \in D, \ \max_{z \in \overline{I}} g_d(z) < 1 \right\}\right) \quad \text{and}\\
		\tilde{\zeta}_+ &= \min\left(\{2\zeta\} \cup \left\{ \min_{z \in \overline{I}} g_d(z) \;\colon\; d \in D, \ \min_{z \in \overline{I}} g_d(z) > \zeta \right\}\right).
	\end{align*}
	Then $\tilde{\zeta}_- < 1$, $\tilde{\zeta}_+ > \zeta$, and $g_d(\overline{I}) \cap (\tilde{\zeta}_-, \tilde{\zeta}_+) = \varnothing$ for all $d \in D$.
	Take a sufficiently small $\varepsilon \in (0,1)$ such that
	\begin{align*}
		(1+\varepsilon) \cdot \widetilde{\delta}_j \cdot \frac{1}{1-\varepsilon} &< \delta_j,  \quad (1+\varepsilon) \cdot \tilde{\zeta}_- \cdot \frac{1}{1-\varepsilon} < 1,\\
		(1-\varepsilon) \cdot \widetilde{\gamma}_j \cdot \frac{1}{1+\varepsilon} &> \gamma_j,\quad (1-\varepsilon) \cdot \tilde{\zeta}_+ \cdot \frac{1}{1+\varepsilon} > \zeta.
	\end{align*}
	Apply \Cref{thm::vn-close-to-un} with $\varepsilon$ to construct $M_\varepsilon$, and let $\widetilde{M} = \max \{M, M_\varepsilon\}$.
	By \Cref{thm::from-un-to-vn,thm::vn-close-to-un}, $u_m/v_m \in (1-\varepsilon, 1+\varepsilon)$ for all $m \ge \widetilde{M}$.
	By the density of $\xi \mu^n$ in~$\torus$, there exist infinitely many $n \ge \widetilde{M}$ with $\xi \mu^n \in I$.
	Pick such $n$, and let $\widetilde{n} = n + d_\ell + 1$.
	Because $\xi \mu^n \in I \subseteq \torus_+$, we have $v_n > 0$, and by \Cref{thm::from-un-to-vn}, $u_n > 0$.
	Recall that $v_{n+d}/v_n = g_d(\xi\mu^n)$ for $d \ge 1$.
	Let $n < m < \widetilde{n}$.
	\begin{itemize}
		\item If $m = n_j \coloneqq n+d_j$ for some $1 \le j \le \ell$, then by \Cref{thm::stacking-intervals}~(a), $v_{n_j} / v_n \in (\widetilde{\gamma}_j, \widetilde{\delta}_j)$. Hence
		\[
		\frac{u_{n_j}}{u_n} = \frac{u_{n_j}}{v_{n_j}} \cdot \frac{v_{n_j}}{v_n} \cdot \frac{v_n}{u_n} < (1+\varepsilon) \cdot \widetilde{\delta}_j \cdot \frac{1}{1-\varepsilon} < \delta_j,
		\]
		and similarly $u_{n_j}/u_n > \gamma_j$.
		Thus $u_{n_j}/u_n \in (\gamma_j,\delta_j) \subseteq (1,\zeta)$.
		\item Otherwise $m - n \in D$.
		If $v_m < 0$, then $u_m < 0 < u_n$ by \Cref{thm::from-un-to-vn}, so $u_m/u_n < 1$.
		If $v_m > 0$, then $v_m / v_n \in (0,\tilde{\zeta}_-] \cup [\tilde{\zeta}_+, \infty)$ as $\xi\mu^n \in \overline{I}$.
		If $v_m/v_n \le \tilde{\zeta}_-$, the same method as before gives $u_m/u_n < \frac{1+\varepsilon}{1-\varepsilon} \tilde{\zeta}_- < 1$.
		If $v_m/v_n \ge \tilde{\zeta}_+$, it gives $u_m/u_n > \frac{1-\varepsilon}{1+\varepsilon} \tilde{\zeta}_+ > \zeta$.
		Thus $u_m/u_n \notin [1,\zeta]$.
	\end{itemize}
	The theorem follows.
\end{proof}

\section{Proof of \Cref{thm::main-DA}}
\label{sec::proof-of-DA-hard-version}
We will proceed similarly to the proof of \Cref{thm::main-DA-function-version}, but we need much stronger technical machinery.
Let $C_1, C_2 > 0$ and $\zeta > 1$ be as in the previous section, and $C_0$ be the maximum of constants of \Cref{lem::baker-distance-from-alpha-n-to-b} applied to $\alpha,\beta$ belonging to $\{\mu,\mu^{-1},1,-1,\im/\zeta, -\im/\zeta\}$.
\begin{lemma}
	\label{thm::log-bound}
	There exists computable $C_3 > 0$ such that for all $1 \le \gamma < \delta \le \zeta$ and $d, n \ge 1$, if $\xi\mu^n \in \mathcal{J}_d(\gamma,\delta)$
	then
	\[
	d < C_3 \log(n+1).
	\]
\end{lemma}
\begin{proof}
	It suffices to prove the statement for $\gamma = 1$ and $\delta = \zeta$.
	By \Cref{thm::int-sizes-upper}, whenever $\xi \mu^n \in \Jcal_d(1,\zeta)$ we have that
	\[
	0 < \Delta(\xi \mu^n, \alpha_d) < \frac{C_1\zeta}{\rho^d}.
	\]
	Since $|\mu| = 1$ and $\alpha_d$ is one of $\pm \im\mu^{-d}$,
	\[
	\Delta(\xi\mu^n, \alpha_d) =  \Delta(\mu^{n+d}, z)
	\]
	where $z$ is one of $\pm \im/ \xi$.
	Applying \Cref{lem::baker-distance-from-alpha-n-to-b} to the right-hand side gives
	\[
	\frac{C_1\zeta}{\rho^d} > \Delta(\xi\mu^n, \alpha_d) > \frac{1}{(n+d)^{C_0}}
	\]
	for all $d, n \ge 1$ such that $\xi \mu^n \in \Jcal_d(1,\zeta)$.
	Hence
	\[
	n  > \frac{(1+\varepsilon)^d}{( C_1\zeta)^{1/C_0}} - d
	\]
	where $\varepsilon \coloneqq \rho^{1/C_0} - 1 > 0$.
	Let $D \ge 1$ and $\widetilde{\varepsilon} \in (0, \varepsilon)$ be such that for all $d \ge D$,
	\[
	\frac{(1+\varepsilon)^d}{( C_1\zeta)^{1/C_0}} - d > (1 + \widetilde{\varepsilon})^d.
	\]
	Then for all $n \ge 1$ and $d \ge D$ such that $\xi \mu^n \in \Jcal_d(1,\zeta)$,
	\[
	n+ 1 > n > (1 + \widetilde{\varepsilon})^d
	\]
	which implies that
	\[
	\frac{1}{\log(1 + \widetilde{\varepsilon})} \cdot \log(n+1) > d.
	\]
	Finally, note that for $d < D$ and $n \ge 1$, regardless of whether $\xi \mu^n \in \Jcal_d(1,\zeta)$ we have that $d < \frac{D}{\log(2)} \log(n+1)$ as $\log(n+1)/\log(2) \ge 1$.
	We can therefore take
	\[
	C_3 = \max \bigg \{\frac{D}{\log(2)} ,  \frac{1}{\log(1+\widetilde{\varepsilon})} \bigg\}. \qedhere
	\]
\end{proof}

Next, we use Baker's theorem to prove a bound on how long it takes for $\seq{\mu^n}$ to fall into a given sub-interval of $\torus$.

\begin{lemma}
	\label{thm::falling-into-an-interval}
	There exists  computable $C_4 > 0$ such that for any non-empty open sub-interval $J \subseteq \torus$, the following holds.
	For any $N \in \nat$, there exists
	\[
	N \le n < N + \bigg( \frac{4\pi}{|J|} \bigg)^{C_4}
	\]
	such that $\mu^n \in J$.
\end{lemma}
\begin{proof}
	Let $l=
	\big\lfloor \frac{2\pi}{|J|}
	\big\rfloor $ and
	consider the intervals $\{J, \ldots, \mu^{l}J\}$ on $\torus$.
	As $(l+1) |J| > 2\pi$, there exist $0 \le m < s \le l$ such that $\mu^mJ$ intersects $\mu^sJ$.
	Let $k = s - m$ and $\kappa = |\Log(\mu^k)|$.
	We have that $0 \le k \le l$ and $\kappa < |J|$.
	Since $\mu$ is not a root of unity, $\mu^m \ne \mu^s$ and hence $k \ge 1$ and $\kappa > 0$.
	We next compute a lower bound on $\kappa$.
	Applying \Cref{lem::baker-distance-from-alpha-n-to-b},
	\[
	\kappa = \Delta(\mu^k,1) > (\max\{2, k\})^{-C_0}
	\]
	for a (computable) constant $C_0 > 0$.
	Since $k \le l < 4\pi/|J|$ and $2 \le 4\pi/|J|$, we have that
	\[
	\kappa > (4\pi /|J|)^{-C_0}.
	\]
	Let $L = \lceil 2\pi/\kappa\rceil$.
	By the lower bound on $\kappa$ above, $L < (4\pi /|J|)^{C}$ for a constant $C > 0$.
	
	Consider the sequence $\seq{z_n}$ of points on $\torus$ that is defined by $z_n =\mu^{N+kn}$.
	We have that $z_{n+1} = \mu^kz_n$ and hence $|z_{n+1}-z_n| < |J|$ for all $n$.
	Moreover, the finite sequence $(z_0, \ldots, z_L)$ winds around $\torus$ at least once.
	Hence there exists
	\[
	0 \le r < L
	\]
	such that $z_{r} \in J$.
	That is, $\mu^{n} \in J$ for $n = N + kr$.
	It remains to observe that $N \le N + kr < N + kL$, and recall the bounds on $k$ and~$L$.
\end{proof}

Recall that our interval theory applies to $\seq{v_n}$, and we translate results to $\seq{u_n}$ using \Cref{thm::from-un-to-vn,thm::vn-close-to-un}.
We next prove a modification of \Cref{thm::main-DA} where we replace $u_n$ with~$v_n$.
Let $\seq{p'_n}$ be the ordering of $\{v_n \ge 0 \colon n \in \nat\}$.

\begin{lemma}
	\label{thm::main-DA-vn-version}
	Let $\ell \ge 1$ and $1 < \gamma_1 < \delta_1 < \cdots < \gamma_\ell < \delta_\ell < \zeta$.
	There exist infinitely many $n$ such that for all $1 \le j \le \ell$,
	\[
	\frac{p'_{n+j}}{p'_n} \in (\gamma_j, \delta_j).
	\]
\end{lemma}
\begin{proof}
	From all $\gamma_j$ and $\delta_j$ construct $d_1, \ldots, d_\ell \ge 1$ as in \Cref{thm::stacking-intervals}.
	Let $I = \Jcal_{d_\ell}(\gamma_\ell, \delta_\ell)$ and $\Dcal = \{0, d_1, \ldots, d_\ell\}$.
	Suppose $n \in \nat$ is such that $\xi \mu^n \in I$.
	Then we have the following.
	\begin{itemize}
		\item Because $I \subseteq \Ical$, $v_n > 0$ and $v_i < v_n$ for all $i < n$.
		\item For all $1 \le j \le \ell$, because $I \subseteq \Jcal_{d_j}(\gamma_j, \delta_j)$, we have that
		\[
		\frac{v_{n+d_j}}{v_n} \in (\gamma_j, \delta_j).
		\]
		\item Consider $1\le k \le d_\ell$ with $k \notin \Dcal$.
		Because $I \cap \Jcal_k(1, \zeta)$ is empty (by \Cref{thm::stacking-intervals} (b)) and $1 < \gamma_j < \delta_j < \zeta$ for all $j$, we have that either $v_{n+k} < v_{n+d}$ for all $d \in \Dcal$, or $v_{n+k} > v_{n+d}$ for all $d \in \Dcal$.
		\item
        Thus, if additionally $\xi \mu^n  \notin \Jcal_k(1,\zeta)$ for all $k > d_\ell$, then the pattern $(v_n, v_{n+d_1},\ldots, v_{n+d_\ell})$ appears in $(p'_m)_{m\in\nat}$.
		In particular, if $v_n = p'_m$, then $\gamma_j < p'_{m+j}/p'_m < \delta_j$ for all $1 \le j \le \ell$.
	\end{itemize}
	Hence it suffices to construct infinitely many $n \in \nat$ such that $\xi \mu^n \in I$, but $\xi \mu^n \notin \Jcal_k(1,\zeta)$ for all $k > d_\ell$.
	Let
	\[
	Y_m = I \setminus \bigcup_{k=d_\ell+1}^m \Jcal_k(1,\zeta).
	\]
	Then by \Cref{thm::log-bound} it suffices to construct infinitely many $n$ such that $\xi \mu^n \in Y_{\lfloor C_3 \log(n+1) \rfloor}$.

	Let $N \in \nat$ be large enough.
	We will construct $N \le n \le 2N$ with the desired property, for which $\xi \mu^n \in Y_{\lfloor C_3 \log(2N+1) \rfloor}$ is thus sufficient as we have that $Y_{m+1} \subseteq Y_m$ for all $m$ and $\lfloor C_3 \log(n+1) \rfloor \le \lfloor C_3 \log(2N+1) \rfloor$ when $n \le 2N$.
	By construction of $I$ (\Cref{thm::stacking-intervals}~(c)), there exists $\tau > 0$ such that $|Y_m| > \tau$ for all $m$.
    Note that $Y_m$ is an interval with $m-d_\ell$ intervals removed, and so $Y_m$ is the union of at most $m-d_\ell + 1\le m$ intervals, as $d_\ell \ge 1$.
	By the pigeonhole principle, each $Y_m$ contains thus an interval of size at least $\frac{\tau}{m}$.
    With $C_4$ as in \Cref{thm::falling-into-an-interval}, there is an $n$ such that $\xi \mu^n \in Y_{\lfloor C_3 \log(2N+1) \rfloor}$ and 
    \begin{equation*}
        2N - \left(\frac{4\pi C_3\log(2N+1)}{\tau}\right)^{C_4} \le n \le 2N.
    \end{equation*}
    For all large enough $N$, we have that $N \ge \left(\frac{4\pi C_3 \log(2N+1)}{\tau}\right)^{C_4}$, showing that $n \ge N$ for all large enough $N$. 
\end{proof}

\begin{proof}[Proof of \Cref{thm::main-DA}]
	Let $M$ be as in \Cref{thm::from-un-to-vn}.
	Suppose we are given $1 < \gamma_1 < \delta_1 \le \gamma_2 < \delta_2 \le \cdots \le  \gamma_\ell < \delta_\ell < \zeta$.
	First construct $\widetilde{\gamma}_j, \widetilde{\delta}_j$ satisfying $\gamma_j < \widetilde{\gamma}_j < \widetilde{\delta}_j < \delta_j$ for $1 \le j \le \ell$, and sufficiently small $\varepsilon > 0$ such that
	\begin{align*}
		(1+\varepsilon) \cdot \widetilde{\delta}_j \cdot \frac{1}{1-\varepsilon} &< \delta_j\\
		(1-\varepsilon) \cdot \widetilde{\gamma}_j \cdot \frac{1}{1+\varepsilon} &> \gamma_j.
	\end{align*}
	Apply \Cref{thm::vn-close-to-un} with $\varepsilon > 0$ to construct $M_\varepsilon$. Then let $\widetilde{M} = \max \{M, M_\varepsilon\}$.
	By \Cref{thm::main-DA-vn-version}, there exist infinitely many $n$ such that for all $j$
	\[
	\frac{p'_{n+j}}{p'_n} \in (\widetilde{\gamma}_j, \widetilde{\delta}_j).
	\]
	Therefore, there exist infinitely many $n_0, \ldots, n_\ell \ge\widetilde{M}$ such that $(v_{n_0}, \ldots, v_{n_\ell})$ appears in $\seq{p'_n}$ and $v_{n_j} >0$ for all $j$.
	By the application of \Cref{thm::from-un-to-vn}, we have that $(u_{n_0}, \ldots, u_{n_\ell})$ appears in $\seq{p_n}$, and $u_{n_j} > 0$ for all $j$.
	Consider such $n_0, \ldots, n_\ell$.

	By the application of \Cref{thm::vn-close-to-un}, for all $0 \le j \le\ell$ we have that $1 - \varepsilon < u_{n_j}/v_{n_j} < 1 + \varepsilon$.
	Therefore,
	\[
	\frac{u_{n_j}}{u_{n_0}} = \frac{u_{n_j}}{v_{n_j}} \cdot \frac{v_{n_j}}{v_{n_0}} \cdot \frac{v_{n_0}}{u_{n_0}} \in \left( \frac{(1-\varepsilon) \widetilde{\gamma}_j}{1+\varepsilon}, \frac{(1+\varepsilon)\widetilde{\delta}_j}{1-\varepsilon}\right).
	\]
	The latter interval is contained in $(\gamma_j, \delta_j)$ by the construction of $\varepsilon$, $\widetilde{\gamma}_j$, and $\widetilde{\delta}_j$.
\end{proof}

\section{Undecidability results for special functions and predicates}
\label{sec::special-functions}

Beyond the previously analysed predicates and functions of dynamical origin, our method for establishing undecidability extends to specific number-theoretic predicates and functions.
As stated in the Introduction, Presburger arithmetic expanded with multiplication and even squaring is undecidable.
Adding fragments of multiplication like relative primeness already leads to definability of multiplication itself and thus back to undecidability~\cite{bes2002survey}.
Bateman, Jockusch, and Woods~\cite{bateman1993decidability} encoded multiplication in Presburger arithmetic expanded with the set of prime numbers, assuming Dickson's conjecture (which we discuss at the end of \Cref{sec:square-free numbers}). 
However, \Cref{cor:Euler and sum of divs undecidable} uses an encoding of arbitrary permutations to achieve undecidability for two expansions of $\langle \nat; 0,1,< \rangle$ with a number-theoretic function, in which definability of multiplication (or, for that matter, addition) is unknown.
In this section, we show that this method applies to many other functions for which it is unclear whether multiplication is definable.

First, let us generalise the notion of a function realising arbitrary permutations.

\begin{definition}
    Let $f \colon \nat \to \nat$ and $\Mb = \langle\nat;0,1,<,f\rangle$.
    We say that $(f, \varphi)$ realises arbitrary permutations, where $\varphi$ is a $(q+1)$-ary predicate in the language of $\Mb$, if the following holds.
    For any permutation $\sigma \colon \{1,\ldots,k\} \to \{1,\ldots,k\}$ there exist $c_1 < \cdots < c_k$ and $\mathbf{x} \in \nat^q$ such that 
    \[
    f(c_{\sigma(1)}) < \cdots < f(c_{\sigma(k)})
    \]
    and for all $n \in [c_1, c_k]$, we have that $n = c_i$ for some $i$ if and only if $\varphi(\mathbf{x}, n)$ holds in $\Mb$.
\end{definition}

Intuitively, the predicate $\varphi$ (which does not depend on $\sigma$) performs a selection of the values of $f(n)$, which then realise the desired permutation.
The following is a straightforward generalisation of \Cref{thm::permutation-to-undec}.

\begin{lemma}\label{lem::permutation plus}
    Suppose $(f, \varphi)$ realises arbitrary permutations, where $\varphi$ is a $(q+1)$-ary predicate.
    Then the first-order theory of $\Mb \coloneqq \langle\nat;0,1,<,f\rangle$ is undecidable.
\end{lemma}
\begin{proof}
    We can use the same encoding as in our proof of \Cref{thm::permutation-to-undec}, with $\mathsf{Seq} \colon \nat^{q+3} \to \nat^*$ and $\mathsf{Rep} \colon \nat^{q+3} \to \nat^*$.
    The only difference is that we have to keep track of which $n$ satisfy $\varphi(\mathbf{x}, n)$.
    For $\mathbf{x} \in \nat^q$ and $c,d,e \in \nat$, let $\mathsf{Rep}(\mathbf{x},c,d,e)$ be the enumeration of
    \[
    \{n \in \nat \colon d < n \le e \text{ and } \varphi(\mathbf{x},n) \textrm{ holds}\}
    \]
    and 
    \[
    \mathsf{Seq}(\mathbf{x},c,d,e)_i = \#\big\{m \in \nat \colon c \le m \le d,\ \varphi(\mathbf{x},m), \text{ and } f(m) < f\big(\mathsf{Rep}(\mathbf{x},c,d,e)_i\big)\big\}.
    \]
    We then define, following the approach of the proof of \Cref{thm::permutation-to-undec},
    \begin{align*}
        &\frm{rep}(\mathbf{x},c,d,e,n) \coloneqq d < n \le e \:\land\: \varphi(\mathbf{x},n)\\
        &\frm{cnst}_0(\mathbf{x},c,d,e,n) \coloneqq \frm{rep}(\mathbf{x},c,d,e,n) \:\land\: \lnot\exists m \in [c,d] \colon (\varphi(\mathbf{x},m) \land f(m) < f(n))\\
        &\frm{cnst}_k(\mathbf{x},c,d,e,n) \coloneqq \frm{rep}(\mathbf{x},c,d,e,n) \:\land\: \\
        &\quad \exists! \{m_1,\ldots,m_k\} \subseteq [c,d] \colon \bigwedge_{i=1}^k \big(\varphi(\mathbf{x},m_i) \land f(m_i) < f(n)\big)\\
        &\frm{succ}(\mathbf{x},c,d,e,n_1,n_2) \coloneqq \frm{rep}(\mathbf{x},c,d,e,n_1) \:\land\: \frm{rep}(\mathbf{x},c,d,e,n_2) \:\land\: \\
        &\quad n_1 < n_2 \:\land\: \forall m \in (n_1,n_2) \colon \lnot\varphi(\mathbf{x},m)\\
        &\frm{inc}(\mathbf{x},c,d,e,n_1,n_2) \coloneqq \frm{rep}(\mathbf{x},c,d,e,n_1) \:\land\: \frm{rep}(\mathbf{x},c,d,e,n_2) \:\land\: \\
        &\quad \exists! m \in [c,d] \colon (\varphi(\mathbf{x},m) \land f(m) \in [f(n_1), f(n_2)))\\
        &\frm{eq}(\mathbf{x},c,d,e,n_1,n_2) \coloneqq \frm{rep}(\mathbf{x},c,d,e,n_1) \:\land\: \frm{rep}(\mathbf{x},c,d,e,n_2) \:\land\: \\
        &\quad \lnot\exists m \in [c,d] \colon (\varphi(\mathbf{x},m) \land f(m) \in [f(n_1), f(n_2)) \cup [f(n_2), f(n_1)))
    \end{align*}
    where $k \in \nat_{\ge 1}$.
\end{proof}

\subsection{The least prime factor function}\label{sec:least prime factor}

Few properties are more integral to integers than their prime factorisation.
However, we will show that even with limited access, we can already obtain undecidability. 
Let $\lpf{n}$ denote the least prime factor of $n$. 
Thus, $\lpf{10} = 2$ and $\lpf{49} = 7$.
\begin{theorem}
    \label{thm::lpf-main}
    Let $q=1$ and $\varphi(x, n) = \lpf{n} \ge x$.
    Then $(\operatorname{lpf}, \varphi)$ realises arbitrary permutations and hence the first-order theory of $\langle\nat;0,1,<,\operatorname{lpf}\rangle$ is undecidable.
\end{theorem}

Let $p_n$ denote the $n$th prime number and let $P_n$ be the product of the first $n$ primes.
Thus, $p_2 = 3$ and $P_n = \prod_{i=1}^n p_i$.
To prove \Cref{thm::lpf-main}, we will need the following version of Mertens' third theorem (see, e.g.,~\cite[Theorem 429]{hardy1979introduction}).
\begin{theorem}\label{thm:Mertens}
    Let $\gamma = 0.57721\ldots$ be the Euler–Mascheroni constant.
    Then,
    \begin{equation*}
        \lim_{n \to \infty} \log(p_n)\prod_{i=1}^n\left(1-\frac{1}{p_i}\right) = \exp(-\gamma).
    \end{equation*}
\end{theorem}
\begin{proof}[Proof of \Cref{thm::lpf-main}]
    Let $k \ge 1$ and $\sigma\colon\{1,\ldots,k\} \to \{1,\ldots,k\}$ be a permutation. 
    We claim we can construct natural numbers $1 < u_1 < \cdots < u_k$ such that 
    \begin{enumerate}
        \item $\lpf{u_{\sigma(1)}} < \lpf{u_{\sigma(2)}} < \cdots <\lpf{u_{\sigma(k)}}$; and
        \item for all integers $u_1 \le u \le u_k$ that are not equal to any $u_i$, we have that $\lpf{u} < \lpf{u_{\sigma(1)}}$.
    \end{enumerate}

    Recall that $\phi$ denotes the Euler totient function.
    For $n \in \mathbb{N}$, exactly $\phi(P_n) = \prod_{i=1}^n (p_i-1)$ numbers among $0,\dots,P_n-1$ are coprime to $P_n$ and thus all of $p_1,\dots,p_n$. 
    By Mertens' third theorem (\Cref{thm:Mertens}) and using that $(p-1)/p = 1-1/p$, we have 
    \begin{equation*}
        \lim_{n\to\infty} \log(p_n)\phi(P_n)/P_n = \lim_{n\to\infty}\log(p_n)\prod_{i=1}^n(1 - 1/p_i) = \exp(-\gamma) > 0.
    \end{equation*}
    As $\lim_{n\to \infty} p_{n+1}/\log(p_n) = +\infty$ (because $p_{n+1} > p_n$ and $ \lim_{m\to\infty} m/\log(m) = +\infty$), we conclude that  
    \begin{equation*}
        \lim_{n\to\infty} p_{n+1} \phi(P_n)/P_n = +\infty.
    \end{equation*}
    
    That is, as $n$ goes to infinity, the expected number of numbers coprime to $P_n$ among $p_{n+1}$ consecutive numbers is unbounded.
    By the limit above and the pigeonhole principle, there are thus natural numbers $m$ and $n$ such that among $m,m+1,\ldots,m+p_{n+1}-1$, there are at least $k$ numbers coprime to $p_1,\dots,p_n$.
	Let $m+m_1 < \ldots < m+m_k$ be the $k$ smallest among them. 
    Using the Chinese Remainder Theorem, the system
    \begin{align*}
        M &\equiv m\pmod{P_n}\\
        M+m_{\sigma(i)} &\equiv 0\pmod{ p_{n+i}} \:\:\text{for $1 \le i \le k$}
    \end{align*} 
    has a solution $M$, which we may take to be greater than $1$. 
    
    We claim that $u_i = M+m_i$ for $1 \le i \le k$ satisfies the claim.
    Indeed, $u_1 < \cdots < u_k$ follows from $m_1<\dots<m_k$.
    As $u_1 \ge M > 1$, we have that $u_1 > 1$.
    For item 1, let $1 \le i \le k$. 
    We have that $p_{n+i}$ divides $u_{\sigma(i)}$ as $u_{\sigma(i)} = M+m_{\sigma(i)} \equiv 0 \pmod{p_{n+i}}$. 
    As $u_{\sigma(i)} \equiv m+m_{\sigma(i)} \pmod{P_n}$ and $P_n$ and $m+m_{\sigma(i)}$ are coprime, $\lpf{u_{\sigma(i)}} \ge p_{n+1}$. 
    If $1 \le j \le k$ satisfies $i\ne j$ and $p_{n+j}$ divides $u_{\sigma(i)}$, then as $p_{n+j}$ divides $u_{\sigma(j)}$ and $u_1 <\cdots < u_k$, we must have that $u_k - u_1 \ge |u_{\sigma(i)}-u_{\sigma(j)}| \ge p_{n+j}$. But this is absurd as by construction, $u_k - u_1 = m_k - m_1 < p_{n+1} \le p_{n+j}$.
    Thus, $p_{n+i}$ is a prime factor of $u_{\sigma(i)}$, and $u_{\sigma(i)}$ has no other prime factor that is at most $p_{n+k}$.
    Therefore, $\lpf{u_{\sigma(i)}} = p_{n+i}$ and so item 1 follows.
    For item 2, let $u_1 \le u \le u_k$ be an integer distinct from $u_1,\dots,u_k$ and set $t = u - (M- m)$. 
    Then $u \equiv t \pmod{P_n}$ and $m+m_1 < t < m+m_k$ while not equal to any $m+m_i$. 
    Hence, by construction, $t$, and thus $u$, is not coprime to $P_n$, having a divisor $p_\ell$ with $\ell \le n$.
    Thus, $\lpf{u} \le p_n < \lpf{u_{\sigma(1)}}$. This proves our claim.

    Now we apply \Cref{lem::permutation plus} with $q=1$ and $\varphi(\mathbf{x}, n) = \lpf{n} \ge x_1$. 
    Then, for $x_1 = \lpf{u_{\sigma(1)}}$, we have that our claim corresponds with the hypothesis of \Cref{lem::permutation plus}. 
\end{proof}

\subsection{Ramanujan tau function}
\label{sec::tau-undec}

Let
$
\Hcal = \{z \in \com \colon \Ima(z) > 0\}.
$
A \emph{modular form} of weight~$k$ is an analytic function $f \colon \Hcal \to \com$ satisfying the following conditions.
\begin{itemize}
	\item For any
	\[
	\Gamma =
	\begin{bmatrix}
		a & b\\
		c & d
	\end{bmatrix} \in \operatorname{SL}_2(\intg)
	\]
	and $z \in \Hcal$,
	\[
	f\bigg(
	\frac{az+b}{cz+d}
	\bigg)
	= (cz+d)^kf(z).
	\]
	The M\"obius transformations $z \mapsto \frac{az+b}{cz+d}$ for $a,b,c,d$ as above are automorphisms of $\Hcal$ and the matrices realising them are exactly the orientation-preserving change-of-basis matrices of the lattice $\intg + \intg z$; thus the condition above states that $f$ behaves well with respect to the symmetries of $\Hcal$.
	\item We can write
	\[
	f(z) = \sum_{n = 0}^\infty a_n q^n
	\]
	where $a_n \in \com$ for all $n$ and $q = e^{\im 2 \pi z}$.
	That is, the Fourier expansion of $f$ in terms of $q$ does not have any negative powers.
	In the most interesting cases, $a_n$ are often real numbers, and sometimes even integers.
\end{itemize}
We say that $f$ is a \emph{cusp modular form} if additionally $a_0 = 0$.
Modular forms (as well as their Fourier coefficients) play a fundamental role in contemporary mathematics, e.g.\ in the theory of elliptic curves and in studying the solutions of Diophantine equations, Wiles' proof of Fermat's Last Theorem~\cite{wiles1995modular} being a prominent example of both.
In this section, we focus on the modular form called the \emph{modular discriminant}, denoted~$\Delta$, and its Fourier coefficients, which are the values of the famous \emph{Ramanujan tau function}; however, our undecidability result can be easily generalised to (coefficients of) a large class of modular forms called \emph{primitive forms}.

The modular discriminant is a cusp modular form of weight~12 (no such forms of weight $2,\ldots,11$ exist) defined~by
\[
\Delta(z) = q \prod_{n=1}^\infty (1-q^n)^{24}
\]
where $q = e^{\im2\pi z}$.
The Ramanujan tau function returns the value of the $n$th Fourier coefficient of $\Delta$, i.e.,
\[
\Delta(z) = \sum_{n=0}^\infty\tau(n)q^n, \quad \tau(0) = 0, \quad \tau \colon \nat \to \intg.
\]
The tau function makes fascinating appearances in a diverse range of fields of mathematics; see, e.g.,~\cite{berndt2025sumssquarestaufunctionramanujans} for an exposition.
It is known to be non-zero infinitely often and multiplicative: $\tau(mn) = \tau(m)\tau(n)$ for any coprime $m,n$.
\emph{Lehmer's conjecture} states that $\tau(n) \ne 0$ for all $n \ge 1$; this has been empirically verified~\cite{zeng2015computation} at least for all $n \le 10^{20}$.
Given $n$ and its factorisation into prime factors in binary, $\tau(n)$ can be computed in polynomial time~\cite[Chapter~15]{couveignes2011computational}.
The following is a specialisation of the main result of~\cite{bilu2018random} concerning coefficients of primitive modular forms.
Write $f(n) = |\tau(n)|$.

\begin{theorem}
	\label{thm::bilu-luca}
	Let $\sigma \colon \{1, \ldots, k\} \to \{1,\ldots,k\}$ be a permutation.
	Assuming Lehmer's conjecture, there exist infinitely many $c \in\nat$ such that
	\[
	f(c+\sigma(1)) < \cdots < f(c+\sigma(k)).
	\]
\end{theorem}
Note that the statement above is just the ``infinitely often'' strengthening of a function realising arbitrary permutations.
Applying \Cref{thm::permutation-to-undec} we obtain the following.
\begin{corollary}
    The first-order theory of $\langle \nat; 0,1,<,n\mapsto|\tau(n)|\rangle$ is undecidable assuming Lehmer's conjecture.
\end{corollary}

A careful examination of the proof of~\cite{bilu2018random} reveals that, even if the tau function were to have infinitely many zeros (which would settle Lehmer's conjecture as being false), then the non-zero values of $|\tau(n)|$ would still realise arbitrary permutations, in the following sense.

\begin{claim}
    Let $q = 0$ and $\varphi(\mathbf{x},n) \coloneqq |\tau(n)| \ne 0$.
    Then $(n \mapsto |\tau(n)|, \varphi)$ realises arbitrary permutations and hence the first-order theory of $\langle \nat; 0,1,<,n\mapsto|\tau(n)|\rangle$ is unconditionally undecidable.
\end{claim}
For reasons of space, we omit the proof, as it would require tracing the whole of~\cite{bilu2018random}.

\subsection{Square-free numbers}\label{sec:square-free numbers}
We next showcase another method of encoding counter machines in our framework. 
Let $r \ge 2$. A number $n \in \mathbb{N}$ is \emph{$r$-power free} if no prime power $p^r$ divides $n$. 
In the special case where $r = 2$, call $n$ \emph{square-free}.
Square-free numbers play a major part in number theory due to the M\"obius function $\mu$~\cite[Section 16.3]{hardy1979introduction}, which is defined by
\begin{equation*}
    \mu(n) = \begin{cases}
        1 & \text{if $n$ is square-free and has an even number of prime divisors}\\
        -1 & \text{if $n$ is square-free and has an odd number of prime divisors}\\
        0 & \text{if $n$ is not square-free.}
    \end{cases}
\end{equation*}
For $r \ge 2$, let $Q_r$ denote the set of $r$-power-free numbers. 
For $r = 2$, Bhardwaj and Tran~\cite{bhardwaj2021additive} explicitly showed that $\langle \nat; <, +, Q_r\rangle$ defines multiplication.
We give an alternative proof of this result to show how it fits into our framework. 
\begin{theorem}\label{thm::square free undecidable}
    For $r \ge 2$, the first-order theory of $\langle \nat; 0,1,<,+,Q_r\rangle$ is undecidable.
\end{theorem}

The following is a simplified version of Theorem 3 of Reuss~\cite{reuss2012pairs}, which is our randomness result that ``everything that can happen, will happen''.
\begin{theorem}\label{thm:Reuss}
    Let $r \ge 2$, $c, k \ge 1$, and $s_1,\dots,s_k$ be distinct integers.
    For a prime $p$, let 
    \begin{equation*}
        \rho(p) = \#\big\{0 \le x < p^r : cx+s_i \equiv 0 \pmod{p^r} \text{ for some } 1 \le i \le k\big\}.
    \end{equation*}
    If $\rho(p) = p^r$ for some prime $p$, there is no $n$ such that $cn+s_1,\dots,cn+s_k$ are simultaneously $r$-power free.
    Otherwise, there are infinitely many such $n$.
\end{theorem}
\begin{lemma}\label{lem::squarefree}
    Let $r \ge 2$, $N \ge 0$, and $(t_i)_{i=1}^N$ be a finite sequence of natural numbers.
    Then there are natural numbers $A$ and $B$ that satisfy the following:
    \begin{enumerate}
        \item there are exactly $2N$ $r$-power free numbers in $[A, B]$, which we label as $a_1<b_1<a_2<b_2<\cdots<a_N<b_N$;
        \item either all $a_i$ are even and all $b_i$ are odd, or all $b_i$ are even and all $a_i$ are odd;
        \item $b_i - a_i = 2t_i + 1$ for $1 \le i \le N$.
    \end{enumerate}
\end{lemma}
\begin{proof} 
For $N = 0$, take $A=B=2^r$.
So assume $N \ge 1$.
We work by induction on $m = 0,1,\dots,N$ to find $a_1'<b_1'<\cdots < a_m'<b_m'$ such that all $a_i'$ and $b_i'$ satisfy the hypothesis of \Cref{thm:Reuss} with $c=1$, $k = 2m$, and the $s_j$ equal to $a_1',b_1',\dots,a_m',b_m'$.
Moreover, $b_i' - a_i' = 2t_i + 1$ for $1 \le i \le m$, all $a_i'$ have the same value modulo $4$, and all $b_i'$ are of the opposite parity.
The case $m=0$ trivially holds and for $m=1$, for any choice of $a_1'$ and $b_1' = a_1'+2t_1+1$, we have that $\rho(p) \le 2 < p^r$ for any prime $p$.

Say $m \ge 1$ and that we found such $a_1'<b_1'<\cdots < a_m'<b_m'$. Then consider $m+1$.
If $k = 2m+2$, $c=1$, and $\rho(p) = p^r$ as in \Cref{thm:Reuss}, we must have that $p^r \le 2m+2$, which holds only for finitely many primes $p$. 
For each of these primes $p \ne 2$, choose a forbidden equivalence class $0 \le d_p < p^r$ not equal to any $a_i'\bmod p^r$ or $b_i' \bmod p^r$ with $i = 1,\dots,m$, which exists by the induction hypothesis.
There is a number $0 \le e_p < p^r$ such that $e_p \not\equiv d_p \not\equiv e_p+(2t_{m+1}+1) \pmod{p^r}$, as we have $p^r \ge 4$ choices for $e_p$ and the equivalence condition is violated for at most two of those.
Using the Chinese Remainder Theorem, take $a_{m+1}'$ such that $a_{m+1}' > b_m'$ and $a_{m+1}' \equiv e_p \pmod{p^r}$ for these finitely many odd prime powers $p^r$ and that $a_{m+1}' \equiv a_m' \pmod{4}$.
Further, let $b_{m+1}' = a_{m+1}'+2t_{m+1}+1$.
Then, $a'_{m+1}$ and $b'_{m+1}$ are not congruent to any $d_p$ modulo $p^r$ for every odd prime $p$ satisfying $p^r \le 2m+2$ and the same holds for all other $a_i'$ and $b_i'$.
Moreover, all $a_i'$ have the same value modulo $4$, while all $b_i' = a_i'+2t_i+1$ have the opposite parity, and so some class modulo $4$ (and thus modulo $2^r$) is not covered.
Hence, taking $k=2m+2$, $c=1$, and the $s_j$ equal to $a_1',b_1',\dots,a_{m+1}',b_{m+1}'$, $\rho(p) < p^r$ for all such primes. 
For all other primes $p$, $\rho(p) < p^r$ since $x \equiv -s_i \pmod{p^r}$ can hold for at most $2m+2 < p^r$ values of $0 \le x < p^r$. 
Thus, $\rho(p) < p^r$ for all primes $p$, and so the hypothesis of \Cref{thm:Reuss} is satisfied, completing the induction.

We now construct $a_1,\dots,a_N,b_1,\dots,b_N$ and $c$ as follows: we choose $c$ such that for any $a_1 \le d \le b_N$ not equal to any $a_i$ or $b_i$, there is a prime power $p^r$ dividing $c$ such that $p^r \mid d$ and $p^r \nmid a_i, b_i$ for $1 \le i \le N$.
So, for each $a_1' \le d' \le b_N'$ unequal to any $a_i'$ or $b_i'$, let $p_{d'}^r > b_{N}'-a_1'+1$ be a distinct prime power and let $c$ be the product of these prime powers. 
Then apply the Chinese Remainder Theorem to find a number $a_1$ such that $a_1 - a_1'+d' \equiv 0 \pmod{p_{d'}^r}$ for all such $d'$.
Set $\delta = a_1-a_1'$, $a_i = \delta+a_i'$ for $2 \le i \le N$ and $b_i = \delta+b_i'$ for $1 \le i \le N$.
Thus, $a_1 < b_1 < \cdots < a_N < b_N$ and $b_i - a_i = 2t_i + 1$ for $1 \le i \le N$.

We check that $\rho(p) < p^r$ for all primes $p$ for this choice of $c$, $k = 2N$, and the $s_j$ equal to $a_1,b_1,\dots,a_N,b_N$.
If $p \nmid c$, there is a number $0 \le x < p^r$ such that $x+a_i' \not\equiv 0 \not\equiv x+b_i' \pmod{p^r}$ for all $1 \le i \le N$ by the induction above.
As $c^{-1}$ modulo $p^r$ exists, take $y = c^{-1}(x-\delta)$: we have that 
\begin{equation*}
    cy+a_i\equiv x+a_i' \not\equiv 0 \not\equiv x+b_i'\equiv cy+b_i \pmod{p^r}
\end{equation*}
for all $1 \le i \le N$.
Thus, $\rho(p) < p^r$.
If $p \mid c$, then by the construction of $c$, we have that $p^r \mid c$ and $p^r > b_N'- a_1'= b_N-a_1$ and so there is at most one multiple of $p^r$ among $a_1,a_1+1,\dots,b_N$.
Since any $a_1 \le d \le b_N$ that is not equal to any $a_i$ or $b_i$ is a multiple of some $p^r$, no $a_i$ or $b_i$ is a multiple of such a $p^r$.
As $c$ is a multiple of these prime powers $p^r$, we thus have that $cx+a_i \not\equiv 0 \not\equiv cx+b_i$ for all $1 \le i \le N$ and $0 \le x < p^r$.
Thus, $\rho(p) = 0 < p^r$.

Thus, \Cref{thm:Reuss} implies that there are infinitely many $n$ such that $cn+a_1,cn+b_1,\dots,cn+a_N,cn+b_N$ are all $r$-power free.
Now for each such $n$, take $A = cn+a_1$ and $B = cn+b_N$.
Then we get $\widetilde{a}_i = cn+a_i=cn+\delta+a_i'$ and $\widetilde{b}_i = cn+b_i=cn+\delta+b_i'$ as in the lemma's statement.
Item 1 follows from $a_1'<b_1'<\cdots<a_N'<b_N'$ and the observation that for all $a_1 \le d \le b_N$ not equal to any $a_i$ or $b_i$, some prime power $p^r$ divides $cn+d$, giving that $cn+d$ is not $r$-power free. 
Item 2 follows from $a_1',\dots,a_N'$ having the same parity and $b_1',\dots,b_N'$ having the opposite parity, while item 3 follows from $\widetilde{b}_i-\widetilde{a}_i=b_i'-a_i'=2t_i+1$.
\end{proof}

\begin{proof}[Proof of \Cref{thm::square free undecidable}]
    We want to use a different type of encoding than before, with $l = 2$ and $m = 1$.
    For $c,d \in \nat$ we define
	\begin{equation*}
	    \mathsf{Rep}(c,d) = (c \le n \le d : Q_r(n) \: \land \: n \equiv c \bmod{2})_{i=1}^N \quad\text{and}\quad \mathsf{Seq}(c,d) = (t_i)_{i=1}^N.
	\end{equation*}
    To define $t_i$, let $n_i$ be the $i$th number in $\mathsf{Rep}(c,d)$ and set 
    \begin{align*}
        \chi(n) &= m \Longleftrightarrow m > n \:\land\: Q_r(m)\:\land\: m \not\equiv n \bmod{2} \:\land\\
        &\quad\forall m' : \big((n < m' < m \:\land\: m \equiv m' \bmod{2})\implies \lnot Q_r(m')\big).
    \end{align*}
    Then we have that
	\[
	t_i = \frac{1}{2}(\chi(n_i)- n_i-1).
	\]
    Let $(t_i)_{i=1}^N$ be given. If it is empty, take $c = d = 0$. Else, apply \Cref{lem::squarefree} to get $a_i$ and $b_i$ and set $c = a_1$ and $d = b_N$. 
    Then $\mathsf{Rep}(c,d) = (a_1,\dots,a_N)$ and $\chi(a_i) = b_i$, giving that indeed, $\mathsf{Seq}(c,d) = (t_i)_{i=1}^N$.

	Next, we define
	\begin{align*}
		&\frm{rep}(c,d,n) \coloneqq c \le n \le d \: \land \:Q_r(n) \: \land \: c \equiv n \bmod{2}\\
		&\frm{cnst}_k(c,d,n) \coloneqq  \frm{rep}(c,d,n) \:\land\: \chi(n) = n +2k+1 \\
		&\frm{succ}(c,d,n_1,n_2) \coloneqq \frm{rep}(c,d,n_1) \:\land\: \frm{rep}(c,d,n_2) \:\land\: n_1<n_2\:\land\: \forall n \in (n_1,n_2) \colon\lnot\frm{rep}(c,d,n) \\
		&\frm{inc}(c,d,n_1,n_2) \coloneqq \frm{rep}(c,d,n_1) \:\land\: \frm{rep}(c,d,n_2) \:\land\: \chi(n_1)+n_2+2=\chi(n_2)+n_1\\
		&\frm{eq}(c,d,n_1,n_2) \coloneqq \frm{rep}(c,d,n_1) \:\land\: \frm{rep}(c,d,n_2) \:\land\: \chi(n_1)+n_2=\chi(n_2)+n_1
	\end{align*}
	where $k \in \nat_{\ge 0}$.
    For $\frm{inc}$, we used that  $\chi(n_1)+n_2+2=\chi(n_2)+n_1$ if and only if $\chi(n_2) -n_2 = \chi(n_1) - n_1+ 2$ as we cannot use subtraction, and $\frm{eq}$ uses the same method.  
    One can easily verify that the conditions of \Cref{thm::how-to-prove-undec} are met, giving the theorem.
\end{proof}

Let us return to the result of Bhardwaj and Tran~\cite{bhardwaj2021additive}, who showed that multiplication is definable for $r = 2$. 
This result also follows from our framework.
Given $N \ge 1$, apply \Cref{lem::squarefree} to the sequence $t_i = 2i^2$, $1 \le i \le N$, to obtain $a_i$ and $b_i$ with $b_i - a_i = 2t_i + 1$. 
This sequence is determined by $t_1 = 2$, $t_2 = 8$, and the recurrence $t_{i+2} = 2t_{i+1} - t_i + 4$, all of which can be expressed in Presburger arithmetic. 
Hence $n$ is of the form $2x^2$ with $x \ge 1$ if and only if $n = (b_i - a_i - 1)/2$ for some $a_i$ and $b_i$ arising in this way. 
This defines the set of squares, and with it multiplication. 
Our counter machine proof does not encode multiplication explicitly, but this shows that our framework captures such undecidability arguments directly.

Bhardwaj and Tran adapted their argument from Bateman, Jockusch, and Woods~\cite{bateman1993decidability}, who showed that Presburger arithmetic expanded with the set of primes $P$ is undecidable, assuming Dickson's conjecture. 
Dickson's conjecture is the analogue of the full result of Reuss~\cite{reuss2012pairs} for primes (only its analogue of \Cref{thm:Reuss} is needed), and it implies well-known conjectures such as the twin prime conjecture. 
We can recover the result of Bateman, Jockusch, and Woods by adapting the proof of \Cref{thm::square free undecidable}: take $b_i - a_i = 4t_i + 2$, $a_i \equiv 1 \pmod 4$, and $b_i \equiv 3 \pmod 4$, and apply the Chinese remainder argument of \Cref{lem::squarefree} with $r = 1$.

Addition is essential for undecidability: in the same paper, Bateman, Jockusch, and Woods showed that the first-order theory of $\langle \nat; 0, 1, <, P\rangle$ is decidable, again assuming Dickson's conjecture. 
We expect the first-order theory of $\langle \nat; 0, 1, <, Q_r\rangle$ to be decidable for $r \ge 2$ as well.

\section{Discussion}
\label{sec::discussion}

Let us briefly discuss the broader implications of our results for the study of decidability of logical theories.
Firstly, we believe that our approach for proving undecidability via simulation of counter machines should work for many more special functions and predicates.
The biggest open problem in this direction is the following.
Write $P$ for the set of all primes.

\begin{problem}
    Is the first-order theory of $\langle\nat;0,1,<,+,P\rangle$ decidable?
\end{problem}
The aforementioned theory is known to be undecidable assuming Dickson's conjecture as it defines multiplication~\cite{bateman1993decidability}. 
By the same conjecture, it simulates counter machines, e.g.\ via a modification of our argument for square-free numbers (see~\Cref{sec:square-free numbers}).
Primes are believed to behave randomly in various specific senses~\cite{Tao2011}, and in our opinion a major open question is whether any known theorem about primes can be transformed into an unconditional proof that the aforementioned structure simulates counter machines.

In this work we focussed on integer LRS with two simple dominant roots that satisfy a non-degeneracy condition, as these were the first large class at the frontier of our understanding.
We believe that, generalising our approach from the one-dimensional torus $\torus$ to higher dimensions, at least for LRS with an irreducible characteristic polynomial, it should be possible to classify decidability of the first-order theories of the structures in \eqref{eq::structures}.
The boundary of decidability and undecidability, however, will not be as clean as it was with our result: with three or more dominant roots, there are examples where Sem\"enov's decidability results apply ``by accident''.

What is the difficulty with reducible LRS?
As an example, consider $u_n = (2 + \im)^n + (2-\im)^n$, which satisfies the recurrence relation $u_{n+2} = 4u_{n+1}-5u_n$.
Let $v_n = 5^n + u_n$, which satisfies $v_n > 0$ for all $n$ and has the characteristic polynomial $p(x) = (x-5)(x^2-4x+5)$ with the single, non-repeated dominant root $\lambda = 5$.
What can we say about the first-order theory of $\langle \nat; +, V\rangle$, where $V= \{v_n\colon n \in \nat\}$?
Define, in $\langle \nat; +, V\rangle$, the predicate $W \subseteq \nat$ by
\begin{multline*}
	x \in W \Leftrightarrow x \ge 0 \:\land\: \exists y_1, y_2 \in V \colon y_1 < y_2 \:\land\: x = y_2 - 5y_1 \:\land\: \forall y_3 \in (y_1, y_2) \colon y_3 \notin V.
\end{multline*}
Then, using that $\seq{v_n}$ is increasing, $W = \{w_n \colon n \in\nat\} \cap \nat$  for 
\[
w_n = u_{n+1} - 5u_n = (\im-3)(2+\im)^n + (-\im-3)(2-\im)^n
\]
which is a non-degenerate integer LRS with exactly two dominant roots.
Therefore, by \Cref{thm::undec-lrs}, the first-order theory of $\langle \nat; +, W\rangle$ is undecidable, which implies the same for $\langle \nat; +, V\rangle$.
More generally, an LRS with a reducible characteristic polynomial ``hides'' more than one LRS inside it: when these can be extracted and we have two or more predicates generated by LRS, undecidability becomes more likely, as illustrated in the work of Hieronymi and Schulz~\cite{hieronymi2022strong}.

We conclude with the following problem for future work.
\begin{problem}
    Let $\Mb = \langle \nat;0,1,<,+,U\rangle$ be any of the structures for which we proved undecidability via simulation of counter machines. 
    Does $\Mb$ define multiplication?
\end{problem}
The seminal negative result of Schulz for powers of two and three~\cite{schulz2023undefinability} shows us a model-theoretic way for approaching this problem.

\bibliography{refs}
\end{document}